\documentclass[11pt,a4paper]{article}
\usepackage{jheppub}

\usepackage[english]{babel}

\usepackage{amsmath,amssymb,amsthm,amsfonts}
\usepackage{exscale}
\usepackage{mathtools}

\usepackage{bbm}
\usepackage{lmodern}

\usepackage{hyperref}

\usepackage{slashed}

\usepackage{graphicx}

\usepackage[all]{xy}
\CompileMatrices

\usepackage{tensor}

\usepackage{ifthen}

\usepackage{tabularx}
\usepackage{booktabs}
\usepackage{shuffle}
\usepackage{mathrsfs}

\numberwithin{equation}{section}

\newcommand{\Q}{
\mathbb{Q}
}
\newcommand{\N}{
\mathbb{N}
}
\newcommand{\Z}{
\mathbb{Z}
}
\newcommand{\R}{
\mathbb{R}
}
\newcommand{\C}{
\mathbb{C}
}
\newcommand{\RP}{
\mathbb{RP}
}

\newcommand{\tp}{
\otimes
}

\newlength{\wurelwidth}

\newcommand{\wurel}[2][=]{\mathrel{\mathop{#1}_{\!\scalebox{0.5}{\makebox[\the\wurelwidth]{#2}}\!}}}

\newcommand{\loops}[1]{\abs{#1}}

\newcommand{\Dim}{D}
\newcommand{\Graph}[2][1.0]{%
\vcenter{\hbox{\includegraphics[scale=#1]{Graphs/#2}}}%
}

\newcommand{\phipol}{\varphi}
\newcommand{\psipol}{\psi}

\newcommand{\hide}[1]{}
\newcommand{\bigo}[1]{\mathcal{O}\left(#1\right)}

\newcommand{\conjugate}[1]{#1^{\ast}}

\newcommand{\dd}[1][]{\mathrm{d}^{#1}}
\newcommand{\restrict}[2]{%
{\left. #1 \right|}_{#2}%
}

\newcommand{\defas}{
\mathrel{\mathop:}=
}

\newcommand{\set}[1]{
\left\{ #1 \right\}
}
\newcommand{\setexp}[2]{
\left\{ #1\!:\ #2 \right\}
}
\newcommand{\abs}[1]{
\left\lvert #1 \right\rvert
}

\newcommand{\SP}{\alpha}
\newcommand{\ep}{a}
\newcommand{\epe}{\nu}
\newcommand{\Hyper}{L}
\DeclareMathOperator{\Li}{Li}
\renewcommand{\L}{L}
\newcommand{\mzv}[2][]{\zeta_{#2}^{#1}}
\newcommand{\cupdot}{\mathbin{\dot{\cup}}}

\DeclareMathOperator{\sdd}{sdd}
\newcommand{\SDD}[2]{\omega_{#1}^{#2}}
\newcommand{\anapartial}[2]{\mathcal{D}_{#1}^{#2}}
\newcommand{\kinematics}{\Theta}

\newcommand{\resultant}[3][]{\left[#2,#3\right]_{#1}}
\newcommand{\discriminant}[2][]{D_{#1}\left( #2 \right)}

\newcommand{\BarObjects}{\mathscr{B}}

\DeclareMathOperator{\ldeg}{deg}

\DeclareMathOperator{\vw}{vw}

\newtheorem{theorem}{Theorem}[section]
\newtheorem{definition}[theorem]{Definition}
\newtheorem{lemma}[theorem]{Lemma}
\newtheorem{corollary}[theorem]{Corollary}

\newtheorem{example}[theorem]{Example}
\newtheorem{remark}[theorem]{Remark}

\title{On hyperlogarithms and Feynman integrals with divergences and many scales}
\author{Erik Panzer}
\emailAdd{panzer@mathematik.hu-berlin.de}
\affiliation{%
Institutes of Physics and Mathematics,
Humboldt-Universit\"{a}t zu Berlin\\
Unter den Linden 6,
10099 Berlin, Germany%
}
\date{\today}

\abstract{
		Hyperlogarithms provide a tool to carry out Feynman integrals in Schwinger parameters. So far, this method has been applied successfully mostly to finite single-scale processes. However, it can be employed in more general situations. 
		
		We give examples of integrations of three- and four-point integrals in Schwinger parameters with non-trivial kinematic dependence, including setups with off-shell external momenta and differently massive internal propagators.
		The full set of Feynman graphs admissible to parametric integration is not yet understood and we discuss some counterexamples to the crucial property of \emph{linear reducibility}. In special cases we observe how a change of variables can restore this prerequisite for direct integration and thereby enlarge the set of accessible graphs.
		
		Working in dimensional regularization, we furthermore clarify how a simple application of partial integration can be used to convert divergent parametric integrands to convergent ones. In contrast to the subtraction of counterterms, this scheme is ideally suited for our method of integration.
}


\arxivnumber{1401.4361}

\begin{document}

\maketitle

\section{Introduction}
	Scalar\footnote{%
	Products of loop momenta in the numerator (in the momentum space representation) yield the same parametric form \cite{Tarasov:ConnectionBetweenFeynmanIntegrals} (see also section 2.3 of \cite{Smirnov:EvaluatingFeynmanIntegrals}) up to shifted powers of $\psi$ and $\phipol$ as well as a further polynomial in the numerator. Therefore such \emph{tensor} integrals are included in our discussion throughout.%
	} %
	Feynman integrals $\Phi(G)$ associated to a Feynman graph $G$ take the form \cite{ItzyksonZuber}
	\begin{equation}
		\Phi(G)
		= \Gamma(\sdd) \cdot
			\left[ \prod_e \int_0^\infty \frac{\SP_e^{\ep_e-1}\dd\SP_e}{\Gamma(\ep_e)}
			\right]
			\psipol^{\sdd - \Dim/2} \phipol^{-\sdd}	
			\cdot
			\delta(H)
		\label{eq:feynman-rules-parametric}%
	\end{equation}
	in Schwinger parameters $\SP_e$ for each edge $e\in E(G)$ and the power $\ep_e$ of the corresponding propagator. 
	The graph polynomials \cite{BognerWeinzierl:GraphPolynomials} are given by sums over all spanning trees $T$ and all spanning two-forests $F$: 
	\begin{equation}%
		\label{eq:graph-polynomials}%
		\psipol 
		= \sum_T \prod_{e\notin T} \SP_e
		\quad\text{and}\quad
		\phipol
		= \sum_{\mathclap{F=T_1\cupdot T_2}} q^2\left( T_1 \right) \prod_{e\notin F}\SP_e
			+	\psipol \sum_e m_e^2 \SP_e,
	\end{equation}
	where $q(T_1) \defas \sum_{v\in T_1} q(v) = -q(T_2)$ denotes the total external momentum entering the tree $T_1$ and $m_e$ the mass of the internal propagator associated to $e$. 

	The Dirac distribution $\delta(H)$ in \eqref{eq:feynman-rules-parametric} projects on an arbitrary\footnote{This freedom of choice is a consequence of \eqref{eq:feynman-rules-parametric} being a projective integral.} hyperplane $\set{H=0}$ which we will always choose as $H=1-\SP_e$ for some fixed edge $e$.
	Denoting the number of loops of $G$ by $\loops{G}$, in $\Dim$ dimensions we declare the superficial degree of divergence as
	\begin{equation}\label{eq:sdd}
		\sdd
		= \sum_e \ep_e
			-
			\frac{\Dim}{2} \loops{G}.
	\end{equation}
	Our strategy is to successively integrate out Schwinger parameters $\SP_{e_1}, \SP_{e_2}, \ldots$ in \eqref{eq:feynman-rules-parametric} following the method of \cite{Brown:TwoPoint} which we implemented in the computer algebra system Maple{\texttrademark} \cite{Maple}. To compute regulated integrals (e.g. $\Dim = 4-2\varepsilon$), we perform the $\varepsilon$-expansion on the integrand of \eqref{eq:feynman-rules-parametric} and integrate out each term individually. 

	This approach requires a convergent integral representation of each term in the expansion, but the immediate form \eqref{eq:feynman-rules-parametric} often turns out to be divergent at the expansion point (e.g. $\varepsilon=0$). In particular this is always the case whenever $G$ contains (infrared or ultraviolet) sub divergences.
	In section \ref{sec:dimreg} we derive a systematic procedure to generate different (but equivalent) parametric integral representations with increased domains of convergence, therefore extending the method of parametric integration to arbitrarily divergent $\varepsilon$-expansions.

	As a consequence, the earlier results for finite single-scale propagator graphs recalled in section \ref{sec:single-scale} generalize to the divergent cases.

	Parametric integration can only be applied to \emph{linearly reducible} graphs $G$, a criterion on the graph polynomials $S_0 \defas \set{\psipol,\phipol}$ which we recall in appendix \ref{sec:linear-reducibility}. 
	The idea is that starting from the integrand $f_0$ of \eqref{eq:feynman-rules-parametric}, for any ordering $e_1,\ldots,e_N$ of edges we can find sets $S_n \in \Q[\SP_{n+1},\ldots,\SP_N]$ of polynomials that describe the possible singularities of the partial Feynman integrals $f_{n+1} \defas \int_0^{\infty} f_n\ \dd \SP_{n+1}$. If each element of $S_n$ is linear in $\SP_{n+1}$, the algorithm of \cite{Brown:TwoPoint} can be applied to compute $f_{n+1}$ in terms of hyperlogarithms.
	These are special classes of multiple polylogarithms and all explicit results in this article will be given in the notation we fix in \ref{sec:hyperlogs}. There we also explain how the final set $S_N$ of this polynomial reduction constrains the symbol of the Feynman integral $\restrict{f_{N-1}}{\SP_N=1} = \Phi(G)$. 

	The main section \ref{sec:non-trivial-kinematics} is a collection of examples of integrals with non-trivial dependence on kinematic invariants $\kinematics = \set{m_e^2,q^2(T),\ldots}$ that are linearly reducible and can thus be integrated parametrically.
	For illustration we supply explicit new results for selected cases, most of which are (due to their volume) not printed but contained in the attached text file only. Further results might be obtained from the author upon request.

	We also point out counterexamples to linear reducibility and show in section \ref{sec:extending-reducibility} that in special cases, changes of variables can allow for parametric integration in spite of the graph not being linearly reducible in the original Schwinger parameters.

\acknowledgments
	I like to thank Christoph Meyer for pointing out to me the importance of graphs with two different internal masses for the NNLO-computation of single top-quark production and checking some analytic results numerically. Bas Tausk provided invaluable help in hinting me to numerous references on selected integrals.
	Furthermore I enjoyed interesting discussions with Johannes Henn who motivated me to investigate integrals with divergences and non-trivial kinematics. He also supplied a series of analytic results that helped me to verify my program.
	Oliver Schnetz explained to me the importance of graphical functions, kindly supplied a list of these and suggested their study in a parametric representation. Finally I like to thank Christian Bogner for many insights into the structure of iterated integrals of many variables and introducing to me the program \cite{BognerWeinzierl:ResolutionOfSingularities} for sector decomposition that was used for some numeric checks in this article. His completely independent program following \cite{BognerBrown:SymbolicIntegration} provided valuable cross-checks of my own implementation.

\section{Single-scale integrals}%
\label{sec:single-scale}

Before studying more complicated examples, let us briefly review integrals with a single scale: $\phipol$ depends only on a single kinematic invariant $\set{s} = \kinematics$ (a mass or external momentum squared) which therefore factors out completely from \eqref{eq:feynman-rules-parametric} as $s^{-\sdd}$.

\subsection{Massless propagators}
	\begin{figure}%
		\begin{gather*}
			\Graph[0.17]{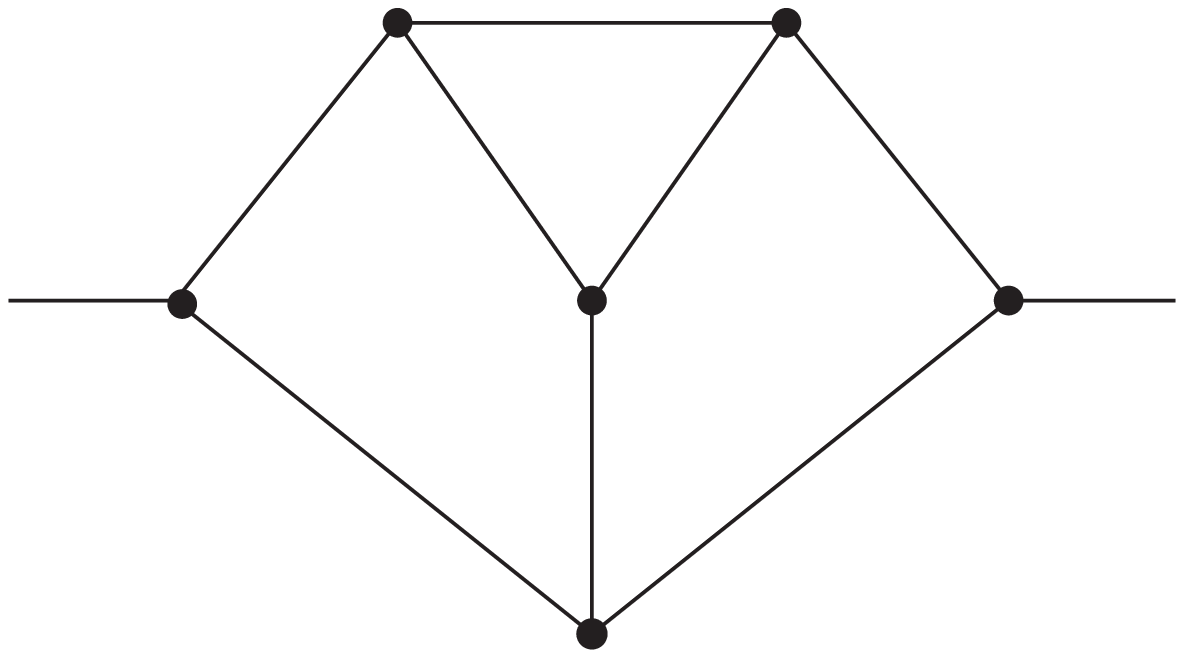}
			\quad
			\Graph[0.2]{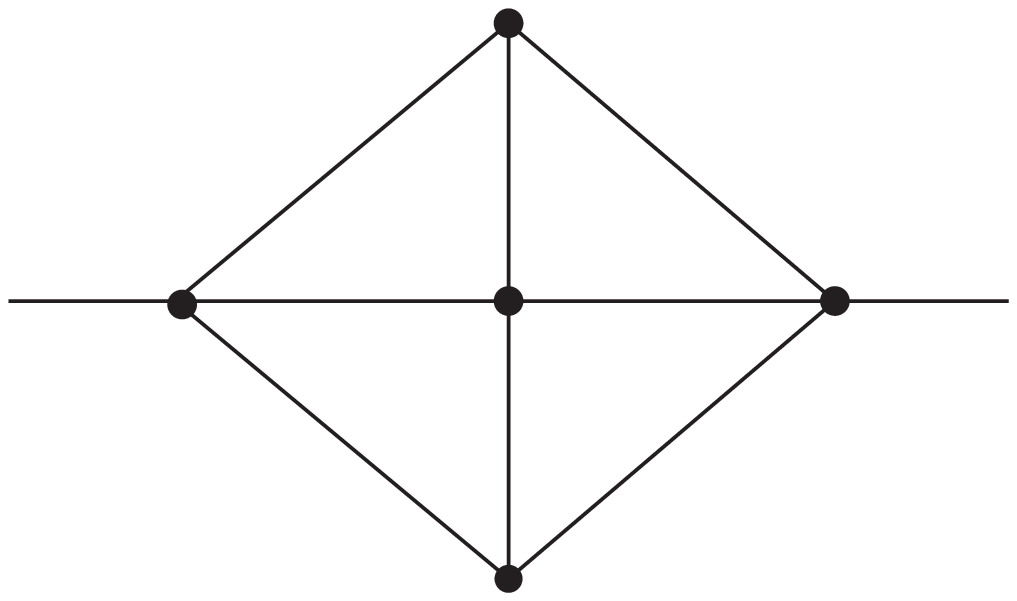}
			\quad
			\Graph[0.17]{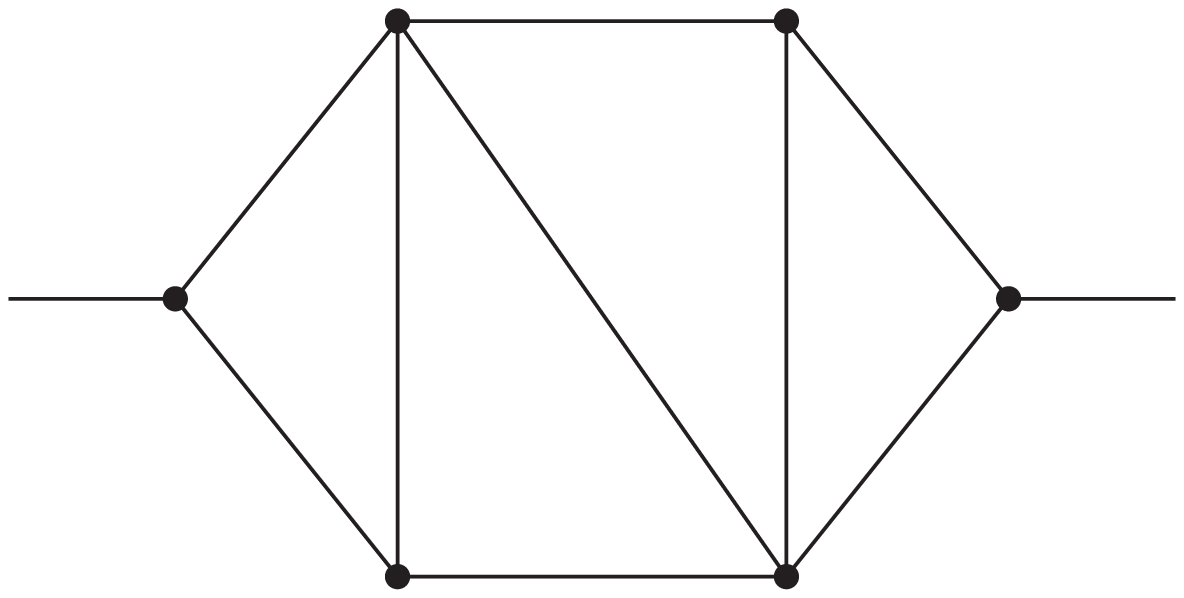}
			\quad
			\Graph[0.3]{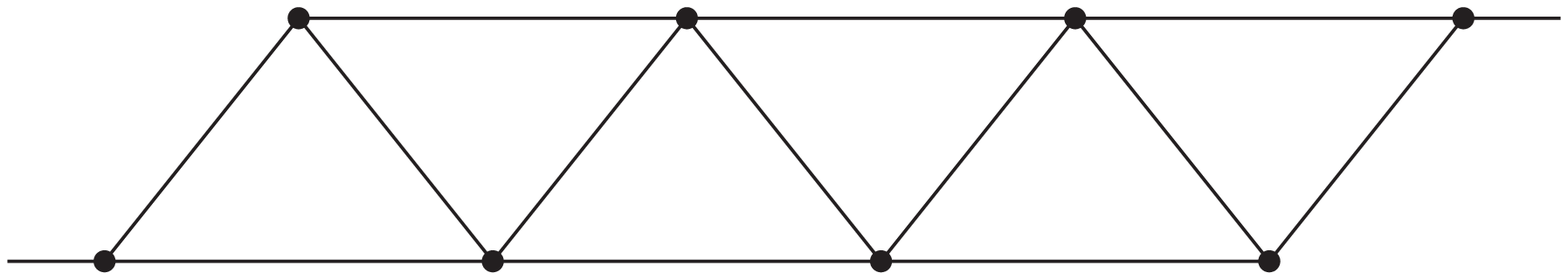}
			\\
			\Graph[0.18]{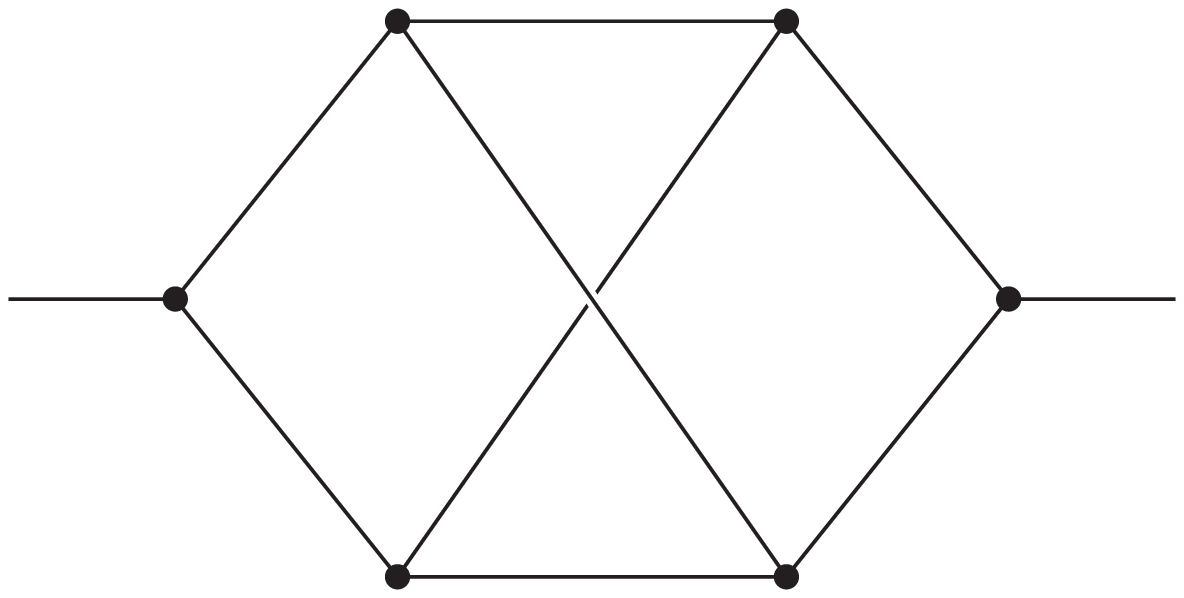}
			\quad
			\Graph[0.18]{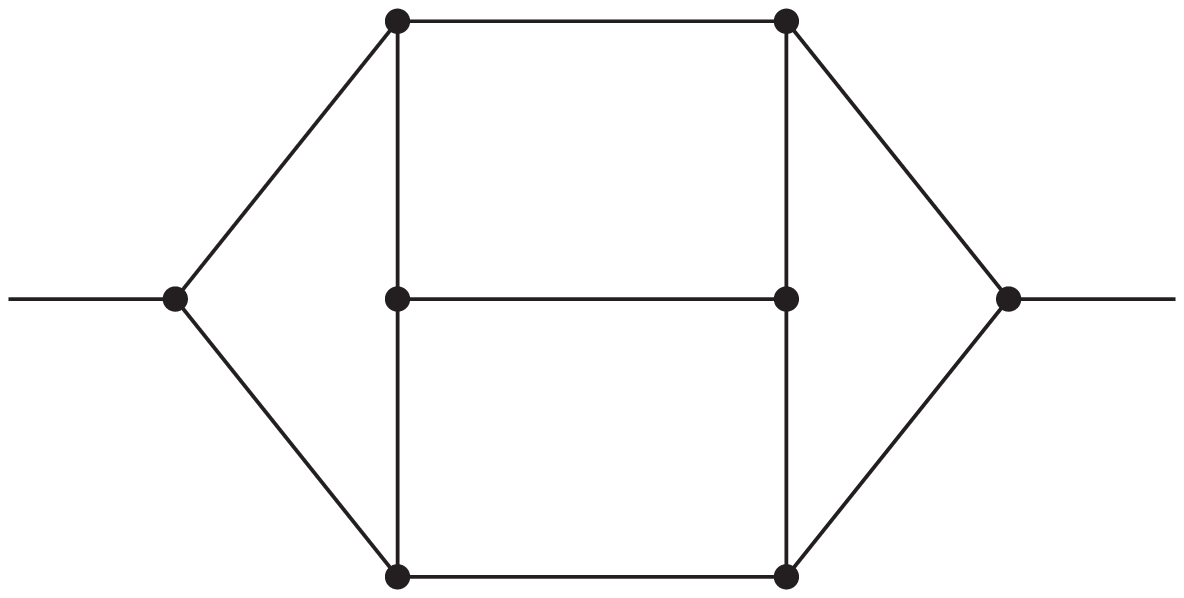}
			\quad
			\Graph[0.18]{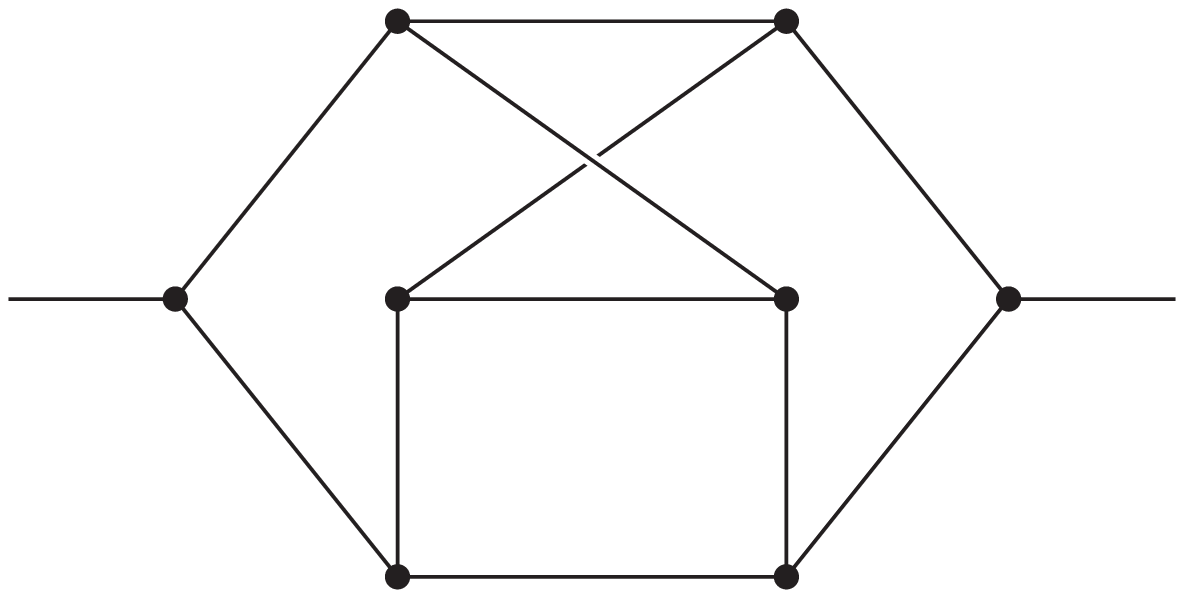}
			\quad
			\Graph[0.18]{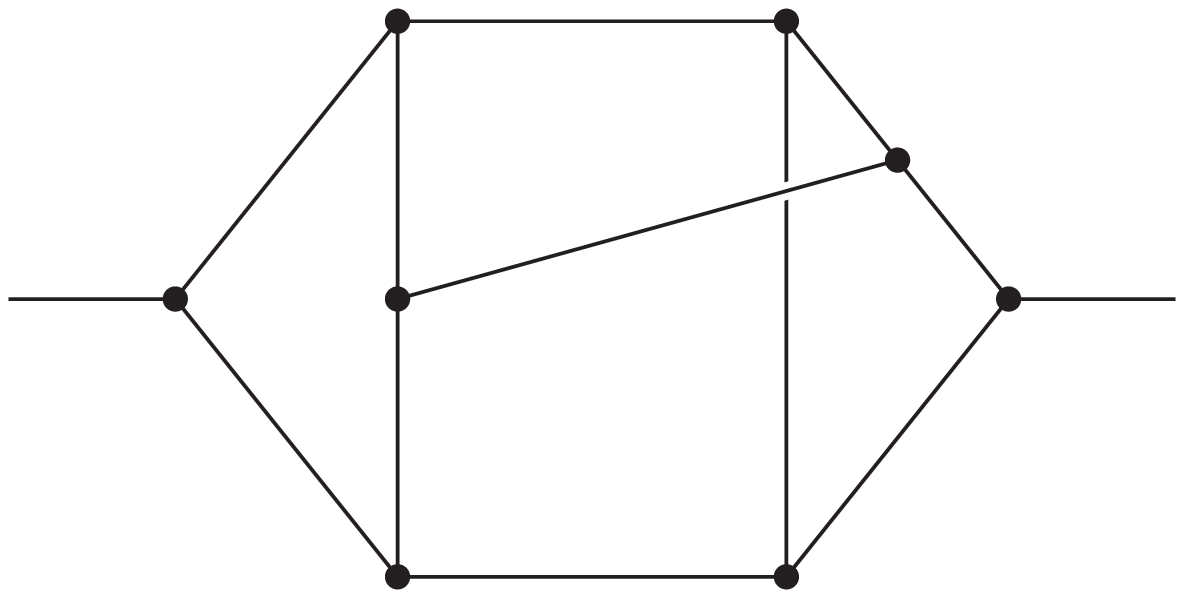}
			\quad
			\Graph[0.18]{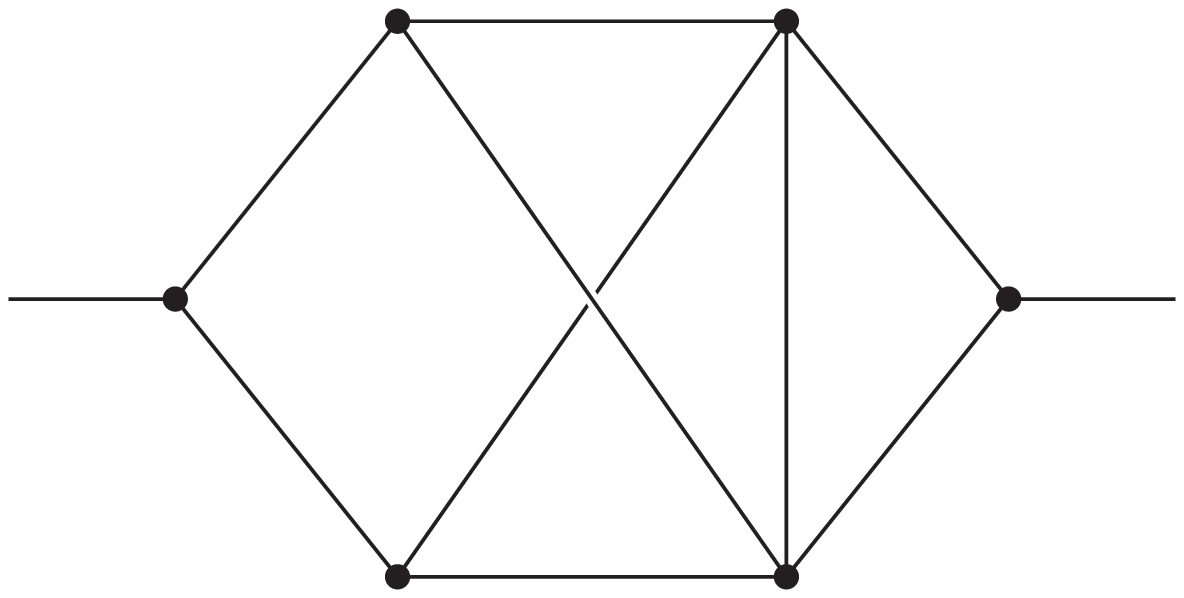}
			\quad
			\Graph[0.18]{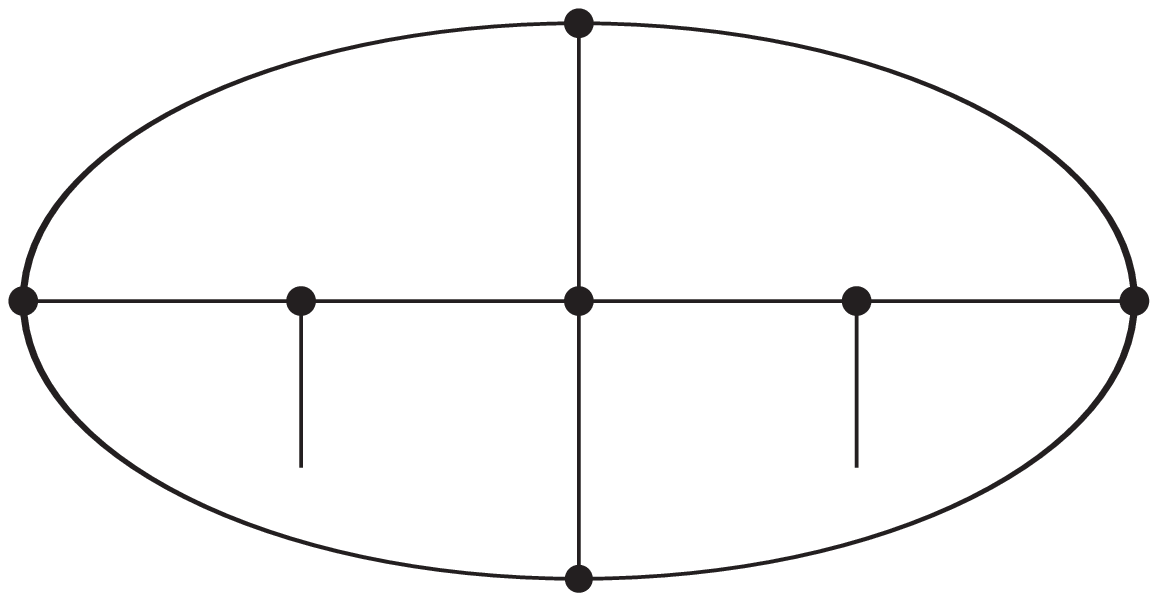}
		\end{gather*}%
		\caption{Examples of massless propagators of vertex-width three (first row) and vertex-width four (second row).}%
		\label{fig:massless-propagators}%
	\end{figure}
	The case of massless graphs $G$ with two external legs (depending on their momenta $\pm p$ through $s = p^2$) is so far the only setup where a non-trivial infinite family of linearly reducible graphs is known to exist by
	\begin{theorem}[positive matrix graphs \cite{Brown:PeriodsFeynmanIntegrals}]\label{theorem:vw-3}
		All massless vacuum (no external momenta) graphs $G$ of \emph{vertex-width} $\vw(G) \leq 3$ are linearly reducible and their $\varepsilon$-expansions are $\Q$-linear combinations of multiple zeta values $\mzv{n_1,\ldots,n_r}$ where $n_1,\ldots,n_r \in \N$ and $n_r > 1$.
	\end{theorem}
	Here $\vw(G) \leq 3$ means that we can order the edges $e_1,\ldots,e_N$ of $G$ such that for all $1\leq n \leq N$, there are at most three vertices of $G$ that touch edges in $\set{e_1,\ldots,e_n}$ and $\set{e_{n+1},\ldots,e_N}$ at the same time. Even though this is a strong\footnote{For example, $\vw(G) \leq 3$ implies planarity of $G$.} constraint on $G$, we like to stress that it holds for infinitely many non-trivial graphs, all of which thus being proven to evaluate to multiple zeta values.
	This theorem extends to massless propagators by glueing the external legs to form a vacuum graph. Examples are shown in figure \ref{fig:massless-propagators}.
	
	Starting at three-loops, graphs with $\vw(G) > 3$ occur (e.g. the second row in the figure) and are therefore not covered by theorem \ref{theorem:vw-3}, but we still have
	\begin{theorem}[vacuum graphs with four or five loops \cite{Panzer:MasslessPropagators}]\label{theorem:4loop-massless-propagators}
		All massless propagators up to four loops are linearly reducible. Their $\varepsilon$-expansions are $\Q$-linear combinations of alternating Euler sums $\Li_{n_1,\ldots,n_r}(\sigma_1,\ldots,\sigma_r)$ where $n_i\in\N$, $\sigma_i^2 = 1$ and $(n_r,\sigma_r)\neq(1,1)$.
	\end{theorem}
	\begin{remark}
		In \cite{Panzer:MasslessPropagators} we carefully stated this result only for convergent $\varepsilon$-expansions, while the influence of sub divergences on the periods was left unclear.
	Taking section \ref{sec:dimreg} into account, we now realize that it holds in full generality and applies to arbitrary dimension $\Dim=\Dim_0 - 2\varepsilon$ ($\Dim_0 \in 2\N$) and propagator powers $\ep_e = n_e + \varepsilon\epe_e$ with $n_e\in\Z$; unaffected by sub divergences (that lead to higher order $\varepsilon$-poles).
\end{remark}
	The first counter-examples to linear reducibility of massless propagators appear at five loops and some are discussed in section twelve of \cite{Brown:PeriodsFeynmanIntegrals}.

\subsection{On-shell propagators with one internal mass}
\begin{figure}\vspace{-7mm}
			\begin{equation*}
				\Graph[0.7]{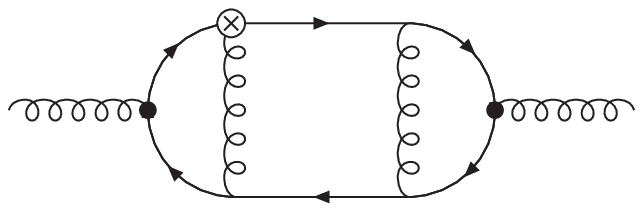}
				\quad
				\Graph[0.7]{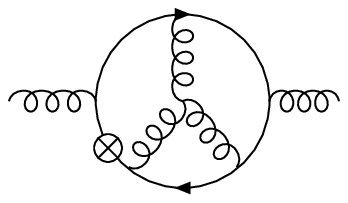}
				\quad
				\Graph[0.7]{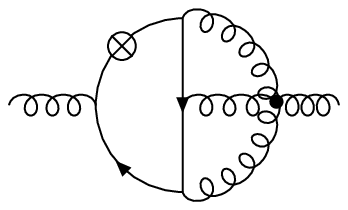}
				\quad
				\Graph[0.7]{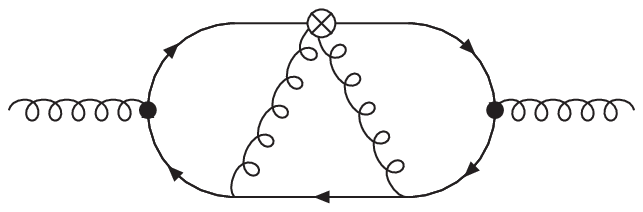}
			\end{equation*}\vspace{-3mm}
			\caption{Linearly reducible topologies with one internal mass (fermion lines) and otherwise massless propagators (including the external momentum $p^2=0$) from \cite{Wissbrock:Massive3loopLadder,Wissbrock:Recent3loopHeavyFlavor,Wissbrock:New3loopHeavyFlavor}.
			The marked vertex represents an operator insertion, its precise form is irrelevant for the polynomial reduction. Note however that the authors aimed for generating functions of all Mellin moments, and then linear reducibility strongly depends on the form of operator.}%
			\label{fig:massive-operator-insertions}%
		\end{figure}
	Another one-scale kinematic setup is given by propagator graphs with light-like external momentum $p^2=0$ but one internal mass $m$. In this case, the second graph polynomial
	\begin{equation}
		\phipol
		=
		m^2 \cdot \psipol \cdot
		\sum_{m_e = m} \SP_e
		\label{eq:no-momenta-phipol}%
	\end{equation}
	splits into polynomials which are themselves linear in each variable, while in general $\phipol$ is irreducible and quadratic in each $\SP_e$ for which $m_e\neq 0$ by \eqref{eq:graph-polynomials}.
	This explains the good linear reducibility despite the presence of many massive edges which was observed in
	\begin{theorem}[\cite{Wissbrock:Massive3loopLadder,Wissbrock:Recent3loopHeavyFlavor,Wissbrock:New3loopHeavyFlavor}]%
		\label{theorem:massive-operator-insertions}%
		The on-shell ($p^2 = 0$) propagators with equally massive internal fermions (and massless gluons) shown in figure \ref{fig:massive-operator-insertions} are linearly reducible.
	\end{theorem}
	Parametric integration was successfully employed in these works to obtain all Mellin-moments of specific operator insertions.

\section{Non-trivial kinematics}%
\label{sec:non-trivial-kinematics}

With increasing number of kinematic invariants, we expect more complicated Feynman integrals and indeed observe in the following a breakdown of linear reducibility at much lower loop orders. Since all reducible graphs evaluate to polylogarithms, known instances like \cite{BlochVanhove:Sunset,AdamsBognerWeinzierl:Sunrise} of elliptic integrals appearing already at two loops are outside the scope of our investigation.

Thus a priori we are restricting ourselves to functions expressible in terms of polylogarithms (with arguments that are algebraic functions of the invariants), and we shall see that not even all of these are linearly reducible in Schwinger parameters.

\subsection{Massless on-shell four-point graphs (two scales)}
\label{sec:massless-onshell-4pt}%
	\begin{figure}
		\begin{gather*}
			\Graph[0.4]{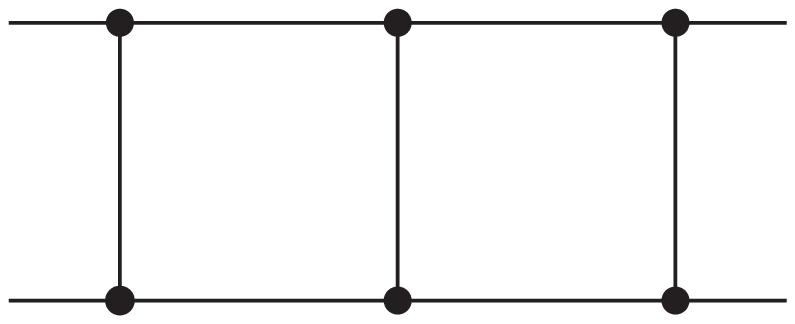}
			\quad
			\Graph[0.4]{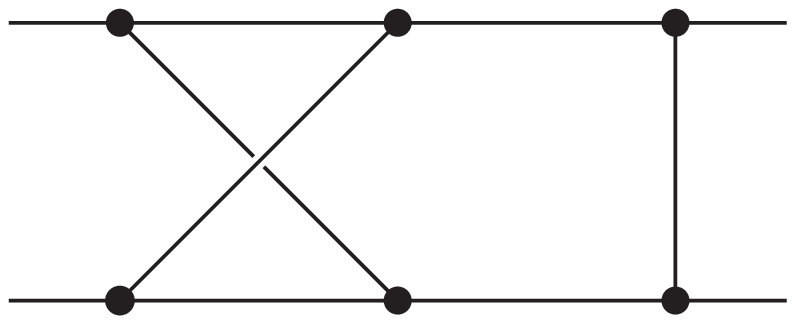}
			\quad
			\Graph[0.4]{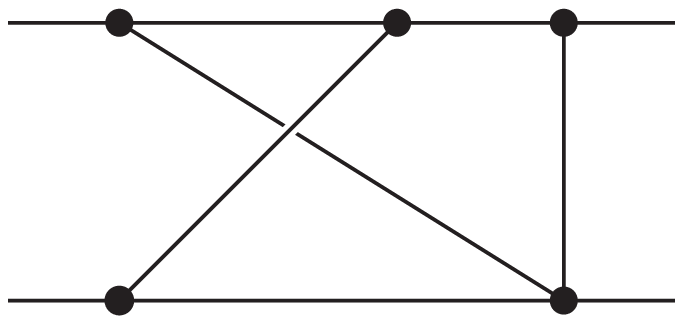}
			\\[2mm]
			\Graph[0.4]{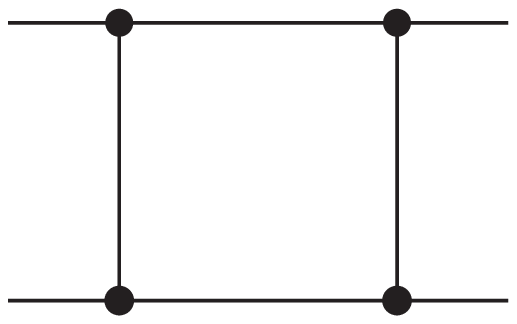}
			\quad
			\Graph[0.4]{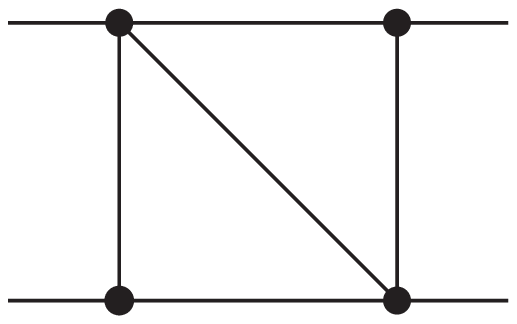}
			\quad
			\Graph[0.4]{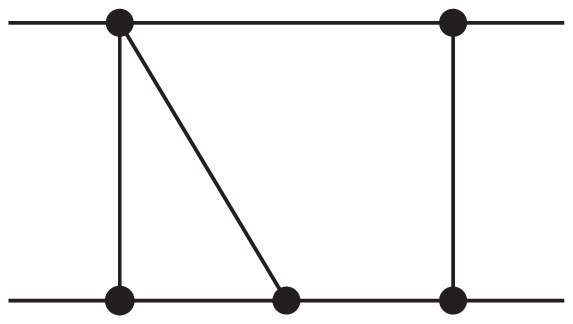}
			\quad
			\Graph[0.4]{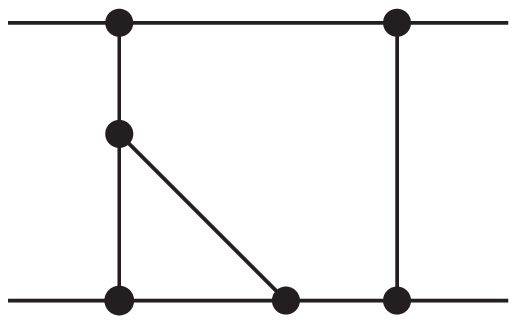}
		\end{gather*}
		\caption{Two-loop four-point functions of theorem \ref{theorem:two-loop-on-shell} without one-scale subgraphs.}
		\label{fig:two-loop-on-shell}
	\end{figure}
	The following result obtained in \cite{BognerLueders:MasslessOnShell} has so far been the only systematic study of linear reducibility for non-trivial kinematics:
	\begin{theorem}%
		\label{theorem:two-loop-on-shell}%
		All massless four-point on-shell graphs ($p_1^2=p_2^2=p_3^2=p_4^2=0$) with at most two loops are linearly reducible. In particular these include those of figure \ref{fig:two-loop-on-shell}.
	\end{theorem}
	This result was expected since these functions were known to evaluate to polylogarithms (even with one leg off-shell \cite{GehrmannRemiddi:DETwoLoopFourPoint}).
	\begin{figure}
		\begin{gather*}
			K_4=\Graph[0.40]{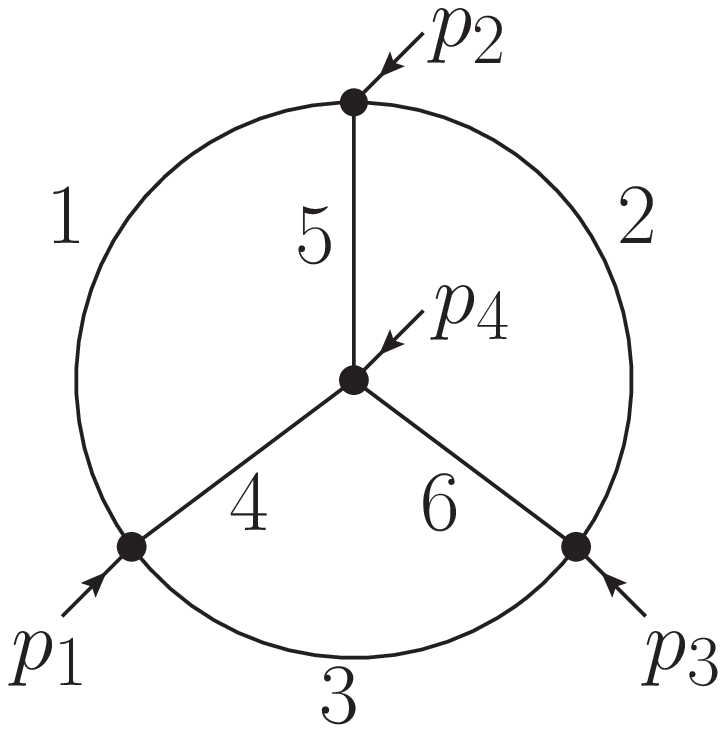}
			\qquad
			G_4=\Graph[0.40]{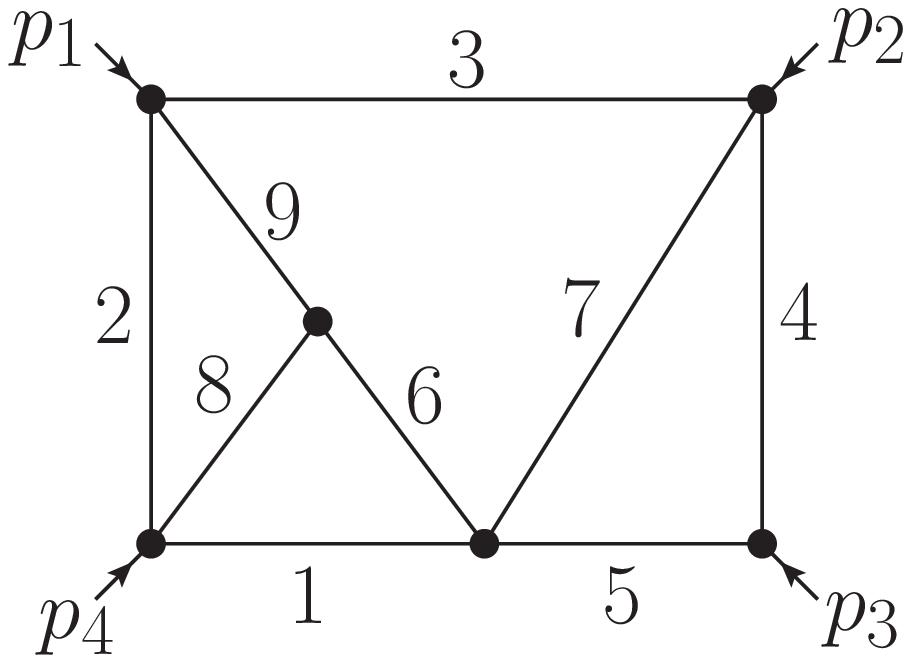}
		\end{gather*}%
		\caption{Massless on-shell four-point graphs: While $G_4$ is linearly reducible (example \ref{example:4pt-4loop}), $K_4$ is not and makes a change of variables necessary (section \ref{sec:K4-variable-change}).}%
		\label{fig:moreloop-on-shell}%
	\end{figure}
	At three loops, counter examples to linear reducibility exist \cite{BognerLueders:MasslessOnShell} like the complete graph $K_4$ of figure \ref{fig:moreloop-on-shell} which was recently evaluated in $\Dim=4-2\varepsilon$ and $\ep_1=\ldots\ep_6=1$ to polylogarithms in \cite{HennSmirnov:K4} using the technique of differential equations.
\hide{
	in the invariants $s=(p_1+p_2)^2$ and $u=(p_1+p_3)^2$:
	\begin{equation}\label{eq:K4}
		\Phi\left( K_4 \right) = 
		\frac{e^{-3\varepsilon\gamma_E} s^{3\varepsilon}}{(1-4\varepsilon)(1-5\varepsilon)} K^{(0)} \left( \frac{s}{t}, \varepsilon \right)
		\quad\text{with}\quad
		K^{(0)} \left( x,\varepsilon \right)
		= \sum_{n\geq -1}
			K_n^{(0)} (x) \cdot \varepsilon^n
	\end{equation}
	where all coefficients $K_{n}^{(0)}$ are polylogarithms (beginning with the constant $K_{-1}^{(0)} = 2\mzv{3}$).
}%
	This proves that a failure of linear reducibility in Schwinger parameters does not prohibit a polylogarithmic result. In fact, section \ref{sec:K4-variable-change} shows how $K_4$ can be integrated parametrically nonetheless.

	While at three loops all planar massless on-shell four-point functions were calculated in \cite{HennSmirnov:PlanarThreeLoop} and the non-planar ones are in progress \cite{HennSmirnov:K4}, results at four loops seem to be very rare.
		From our above observations it seems plausible that at least some of them are linearly reducible.
		
	\begin{example}\label{example:4pt-4loop}
		The linearly reducible graph $G_4$ of figure \ref{fig:moreloop-on-shell} can be integrated in Schwinger parameters along the sequence $1,2,8,6,9,7,5,4$ of edges (setting $\SP_3 = 1$). The final set of polynomials $S_9=\set{s+u}$ in the reduction proves that all coefficients $f_n$ of the $\varepsilon$-expansion
		\begin{equation}%
			\label{eq:4pt-4loop-onshell}%
			\Phi\left( G_4 \right)
			= \frac{\Gamma(1+4\varepsilon)}{s^{1+4\varepsilon}}
			\sum_{n=-1}^{\infty} f_n\left(\frac{s}{u}\right) \cdot \varepsilon^n
			\quad\text{where}\quad
			s = (p_1+p_2)^2,
			u = (p_1+p_4)^2
		\end{equation}
		are harmonic polylogarithms $f_n \in L\left( \set{0,-1} \right)$ of $x\defas \frac{s}{u}$. 
		These are special hyperlogarithms introduced in \cite{RemiddiVermaseren:HarmonicPolylogarithms} and abbreviated $H_{n_1,\ldots,n_r} \defas \Hyper_{\underline{n_1},\ldots,\underline{n_r}}(x)$ for indices representing words $\underline{0} \defas \omega_{0}$ and $\underline{\pm n} \defas \mp \omega_0^{\abs{n}-1} \omega_{\pm 1}$, e.g. $H_{-2} = \Hyper_{\omega_0\omega_{-1}} (x) = \Li_2(-x)$. 
		Explicitly we computed
		\begin{align}%
			f_{-1} &=
 - \tfrac{79}{70}\mzv[3]{2}H_{{-1}}
 - \mzv{3}
\left(
15\mzv{2}H_{{-1,-1}}
 - 9\mzv{2}H_{{-1,0}}
 - H_{{-1,-2,-1}}
 + H_{{-1,-1,-2}}
 + 6H_{{-1,-1,0,0}}
\right)
\nonumber\\&\quad
 - 6\mzv[2]{3}H_{{-1}}
 - \tfrac{3}{2}\mzv{5}
\left(
11H_{{-1,-1}}
 - 5H_{{-1,0}}
\right)
 - \tfrac{3}{10}\mzv[2]{2}
\left(
H_{{-1,-2}}
 - 17H_{{-1,-1,0}}
 - 10H_{{-1,-1,-1}}
\right)
\nonumber\\&\quad
 - \mzv{2}
\Big(
 H_{{-1,-2,0,0}}
 - 2H_{{-1,-1,-2,0}}
 + 3H_{{-1,-1,-2,-1}}
 - H_{{-1,-1,-1,0,0}}
 + 6H_{{-1,-1,-3}}
\nonumber\\&\qquad\qquad
 - 3H_{{-1,-2,-1,-1}}
 - 2H_{{-1,-1,0,0,0}}
\Big)
 + H_{{-1,-2,-1,0,0,0}}
 - H_{{-1,-1,-2,-1,0,0}}
\nonumber\\&\quad
 + H_{{-1,-1,-2,0,0,0}}
 - 2H_{{-1,-1,-3,0,0}}
 + H_{{-1,-2,-1,-1,0,0}}
			\label{eq:4pt-4loop-1}%
		\end{align}
		while we provide $f_0$ in the attached file. With FIESTA \cite{SmirnovTentyukov:FIESTA} we obtained the approximation
		\begin{equation}
			\frac{\Phi\left( G_4 \right)}{\Gamma(1+4\varepsilon)}
			\approx
			-219.35 \varepsilon^{-1}
			-3626.82
			+\bigo{\varepsilon}
			\quad\text{at}\quad
			(s,u) 
			=(s_0,u_0)
			\defas
			\left(\frac{1}{2}, \frac{1}{5} \right)%
			\label{eq:4pt-4loop-numeric}%
		\end{equation}%
		which serves a successful independent check of our analytic result since it produces the (exact) first digits $-219.4440\ldots \varepsilon^{-1} - 3630.1071\ldots + \bigo{\varepsilon}$ at $(s_0, u_0)$.
	\end{example}

\subsection{Off-shell massless vertices (three scales)}
	Let us consider graphs with massless internal propagators ($m_e=0$) and three external momenta as shown in figure \ref{fig:vertices-2loops}.
	\begin{figure}
		\begin{equation*}
			\Delta_1
			\defas
			\hspace{-7mm}\Graph[0.4]{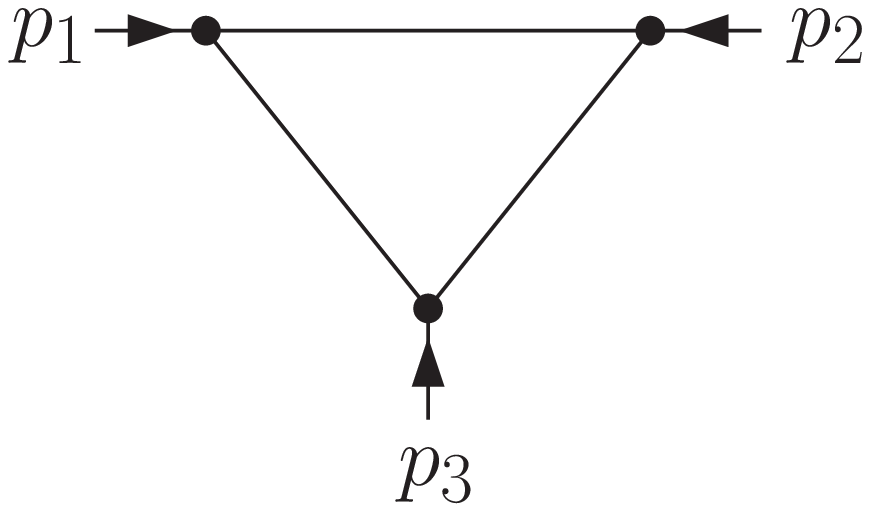}
			\quad
			\Graph[0.4]{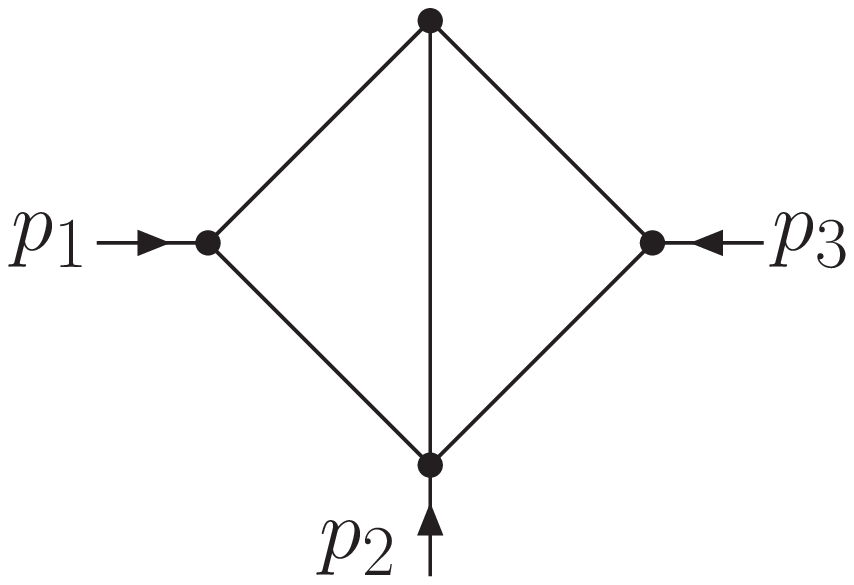}
			\quad
			\Graph[0.4]{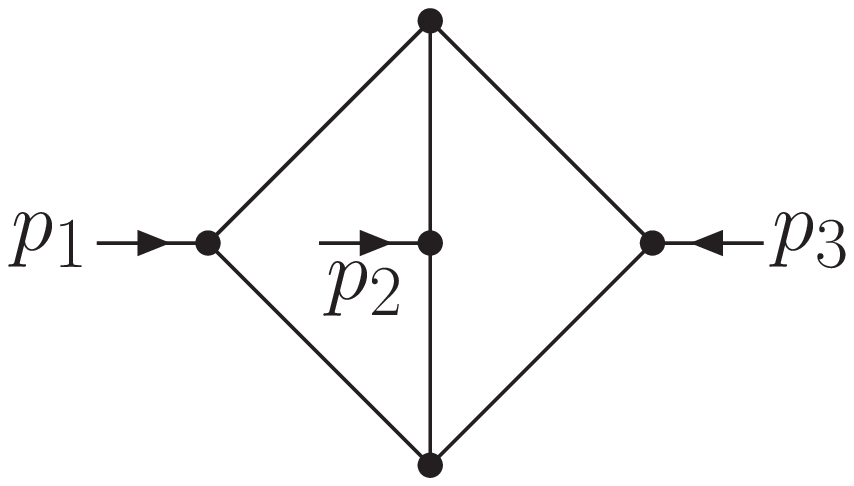}
			\quad
			\Graph[0.4]{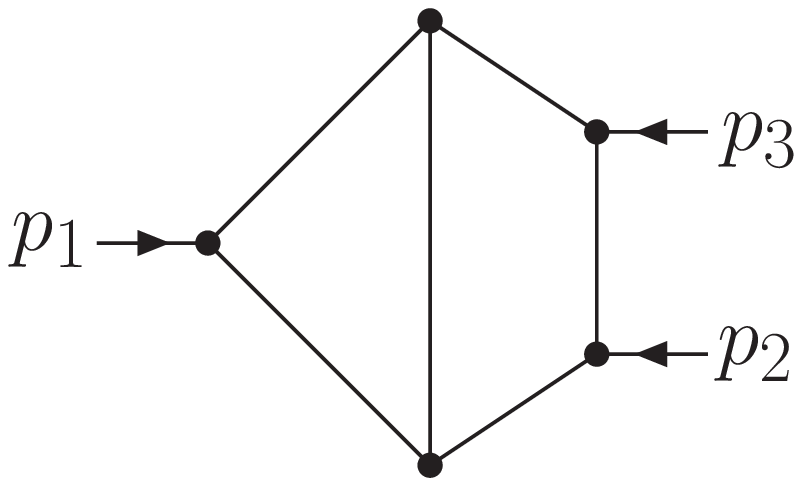}
		\end{equation*}
		\caption{All massless three-point graphs with one or two loops and without one-scale subgraphs (massless propagator insertions). Results are given in \cite{ChavezDuhr:Triangles}.}
		\label{fig:vertices-2loops}
	\end{figure}
	At two loops their linear reducibility was observed in \cite{ChavezDuhr:Triangles} and by integration of Schwinger parameters explicit results up to weight four were obtained.
	These are most conveniently expressed in terms of auxiliary complex\footnote{%
		The Euclidean region $p_1^2,p_2^2,p_3^2>0$ corresponds to positive values of $z\bar{z}$ and $(1-z)(1-\bar{z})$. This means either complex conjugate $\bar{z} = \conjugate{z}$ (when $\lambda<0$) or independent real $z,\bar{z}\in\R$ (when $\lambda>0$); cf. \cite{ChavezDuhr:Triangles}.}%
	variables $z,\bar{z}$ such that the square-root of the K\"{a}ll\'{e}n function $\lambda$ becomes rational:
	\begin{equation}\label{eq:triangle-kinematics}\begin{split}
		p_2^2
		&= p_1^2 \cdot z \bar{z}
		\qquad\text{and}\qquad
		p_3^2 
		= p_1^2 \cdot (1-z) (1-\bar{z})
		, \qquad\text{such that}
		\\
		(z-\bar{z})^2
		&= 
		\lambda \defas
		p_1^2 + p_2^2 + p_3^2 - 2p_1 p_2 - 2 p_1 p_3 - 2p_2 p_3.
	\end{split}\end{equation}
	\begin{example}\label{example:vertex-oneloop}
	The triangle $\Delta_1$ expands near $\Dim=4-2\varepsilon$ with $\ep_1=\ep_2=\ep_3=1$ as
	\begin{align*}
		\Phi\left(\Delta_1\right)
		&=\frac{\Gamma(1+\varepsilon)}{z-\bar{z}} \cdot p_1^{-2(1+\varepsilon)} \cdot \sum_{n=0}^{\infty} f_n(z,\bar{z}) \varepsilon^n
		\quad\text{with the leading order given by}
		\\%
		f_0
		&=4 i \Im\left\{ \Li_2(z)+\ln\abs{z} \cdot \ln(1-z) \right\}
		,\quad\text{the Bloch-Wigner dilogarithm}.
		\end{align*}
\end{example}
	Up to two loops, the functions (like $f_n$ in example \ref{example:vertex-oneloop}) occurring in the $\varepsilon$-expansions have symbols with letters drawn from the alphabet $\Sigma_{\Delta} \defas \set{z,\bar{z},1-z,1-\bar{z},z-\bar{z}}$. These where first studied in \cite{ChavezDuhr:Triangles} and generalize the single-valued multiple polylogarithms of \cite{Brown:Uniformes}. 
	Running the polynomial reduction algorithm proved
	\begin{theorem}
		\label{theorem:vertices}%
		All massless three-point graphs with up to three loops (some examples are depicted in figure \ref{fig:vertices-3loops}) are linearly reducible.
	\end{theorem}
	\begin{figure}
		\begin{gather*}
			\Delta_{3,1}
			\defas
			\Graph[0.3]{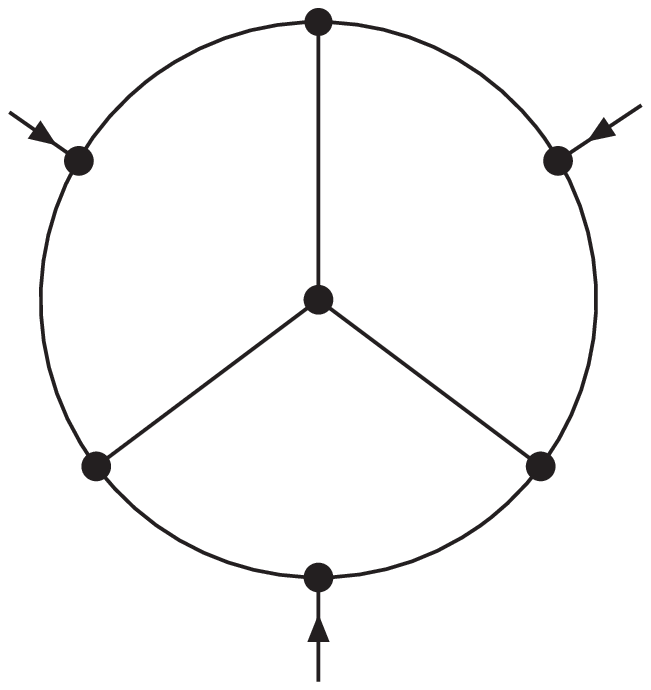}
			\qquad
			\Delta_{3,22}
			\defas
			\Graph[0.3]{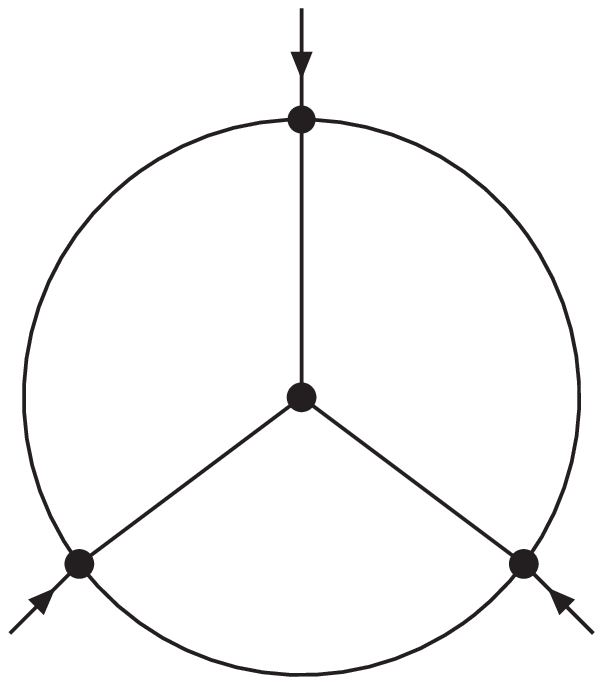}
			\qquad
			\Delta_{3,14}
			\defas
			\Graph[0.3]{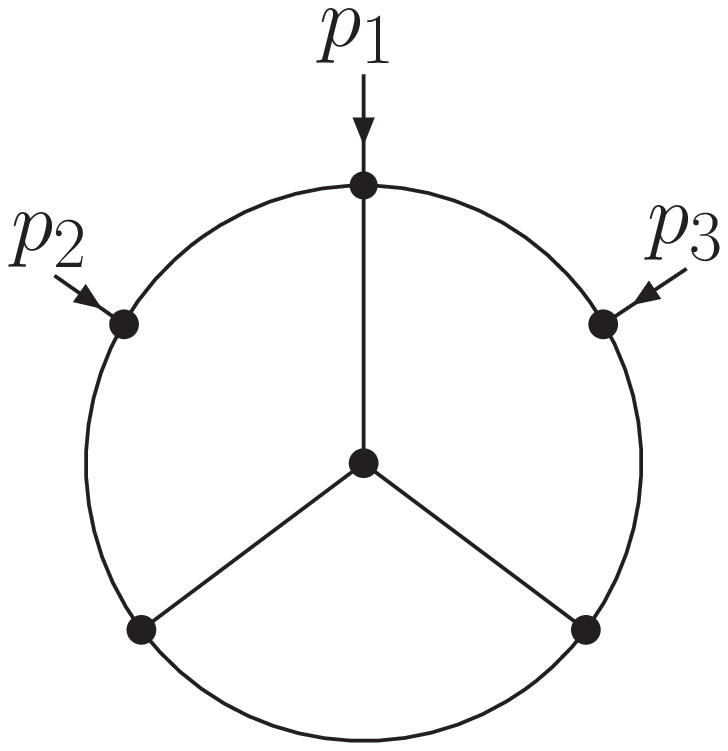}
			\\
			\Delta_{3,5}
			\defas
			\Graph[0.4]{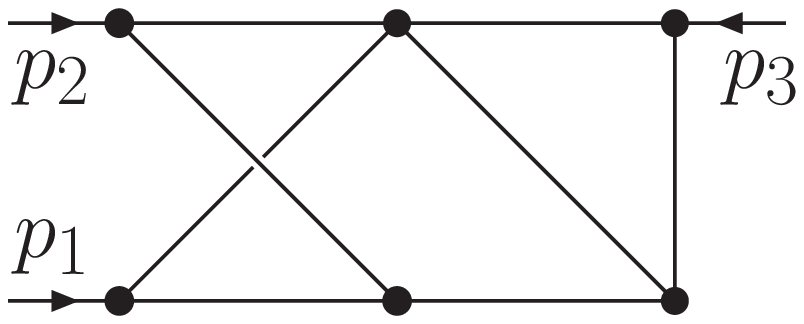}
			\qquad
			\Delta_{3,20}
			\defas
			\Graph[0.4]{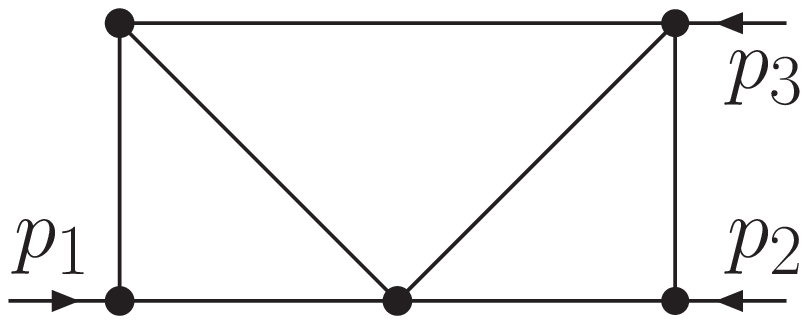}
			\qquad
			\Delta_{3,2}
			\defas
			\Graph[0.3]{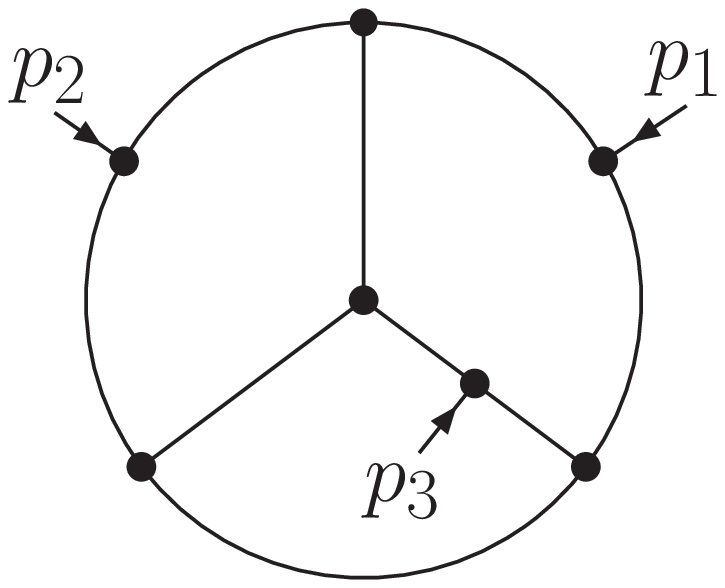}
		\end{gather*}
			\caption{These three-loop three-point graphs are discussed in example \ref{example:3loop-vertices}.}%
		\label{fig:vertices-3loops}%
	\end{figure}
	The final sets $S_N$ of the polynomial reduction provide the alphabets of the symbols. We found that these always contain the set $\Sigma_{\Delta}$ familiar from two loops, but in some cases also the additional letters $z\bar{z}-1$, $z+\bar{z}-1$ and $z\bar{z}-z-\bar{z}$ occur.

	We computed these functions in $\Dim=4-2\varepsilon$ dimensions with unity propagator powers $\ep_e=1$ for all edges $e\in E$ and performed checks exploiting symmetry properties, known results in single-scale limits and numeric evaluations.
	However, the length of the results and the rich structure of the occurring polylogarithms suggests a detailed and separate discussion elsewhere. We provide some selected data for the graphs of figure \ref{fig:vertices-3loops} in
	\begin{example}%
		\label{example:3loop-vertices}%
		Three-point functions can have very different complexity: The simplest examples of figure \ref{fig:vertices-3loops} evaluate in leading order to rational functions like
		\begin{equation}\label{eq:v3_1v3_14}%
			\Phi\left( \Delta_{3,1} \right)
			= \frac{20\mzv{5}}{p_1^2 p_2^2 p_3^2} + \bigo{\varepsilon}
			\quad
			\Phi\left( \Delta_{3,14} \right)
			= -\frac{2 \mzv{3}}{p_2^2 p_3^2 \varepsilon} + \bigo{\varepsilon^0}
			\quad
			\Phi \left( \Delta_{3,22} \right)
			= \frac{2\mzv{3}}{\varepsilon} + \bigo{\varepsilon^0},
		\end{equation}
		or to hyperlogarithms $\L_w\defas \L_w(z)$ and $\bar{\L}_w \defas \L_{w'}\left( \bar{z} \right)$ with $w,w'\in\set{0,1}^{\times}$ not involving the symbol letter $z-\bar{z}$ (these are called SVMP in \cite{Schnetz:GraphicalFunctions} or SVHPL in \cite{DDEHPS:LeadingSingularitiesOffShellConformal}), e.g.
		\begin{align}
			\Phi\left( \Delta_{3,20} \right)
&=
\frac{p_1^{-2}}{z-\bar{z}} \Big\{
	\mzv{3} \left(
		4\,\bar{L}_{0,1,1}
		-4\,L_{0,1,1}
		+6\,L_{0,1}\bar{L}_{0}
		-6\,L_{0}\bar{L}_{0,1}
		-6\,\bar{L}_{0,1,0}
		+6\,L_{0,1,0}
 \right)
\nonumber\\&\quad
+L_{0,1,1,0,1,0}
-\bar{L}_{0,1,1,0,1,0}
+\bar{L}_{0,1,0,1,1,0}
-L_{0,1,0,1,1,0}
+L_{0,1,1}\bar{L}_{0,1,0}
-L_{0,1,0}\bar{L}_{0,1,1}
\nonumber\\&\quad
+L_{0}\bar{L}_{0,1,0,1,1}
-L_{0,1,0,1,1}\bar{L}_{0}
+L_{0,1,1,0,1}\bar{L}_{0}
-L_{0}\bar{L}_{0,1,1,0,1}
\nonumber\\&\quad
+L_{0,1}\bar{L}_{0,1,0,1}
-L_{0,1,0,1}\bar{L}_{0,1}
+L_{0,1,1,0}\bar{L}_{0,1}
-L_{0,1}\bar{L}_{0,1,1,0}
\Big\}
+ \bigo{\varepsilon}
.%
			\label{eq:v3_20}%
		\end{align}
		In contrast, already the leading order of $\Phi\left( \Delta_{3,2} \right)$ needs the letter $z-\bar{z}$ and the subleading contribution to $\Phi\left( \Delta_{3,5} \right)$ further employs $z\bar{z}-1$ and has $ 2348 $ different terms $\L_w \cdot \bar{\L}_{w'}$. These expansions, including $\Phi\left( \Delta_{3,22} \right)$ up to order $\varepsilon^3$, can be found in the ancillary file.
	\end{example}
	Let us stress that linear reducibility of course is retained upon specializing $p_i^2=0$ to be light-like for one or two of the external momenta, corresponding to (possibly singular) limits $z\rightarrow 0,1,\infty$.
	In particular, combining the remarks of section \ref{sec:dimreg} with theorem \ref{theorem:vertices} implies that all three-loop \emph{form-factor integrals} as studied for example in \cite{GehrmannHeinrichHuberStuderus:FormFactorInsertions,HeinrichHuberKosowerSmirnov:NinePropagatorFormFactors} can be integrated parametrically.

	Regarding the quickly growing number of graphs at even higher loop orders, but also from a purely conceptual viewpoint, a combinatorial criterion (in the spirit of theorem \ref{theorem:vw-3}) on a three-point graph that at least suffices to deduce linear reducibility (without the need of running the polynomial reduction algorithm) in some cases is highly desirable and in progress. For now let us only remark that reducible graphs also exist at higher loop orders.
	\begin{example}%
		\label{example:4loop-vertex}%
		The non-planar four-loop three-point graph $\Delta_4$ of figure \ref{fig:vertices-moreloops} is linearly reducible, its leading order contribution in $D=4-2\varepsilon$ is supplied in the attached file and has a symbol with letters $\Sigma_{\Delta}$.
	\end{example}
	\begin{figure}
		\begin{equation*}
			\Delta_4 \defas \Graph[0.4]{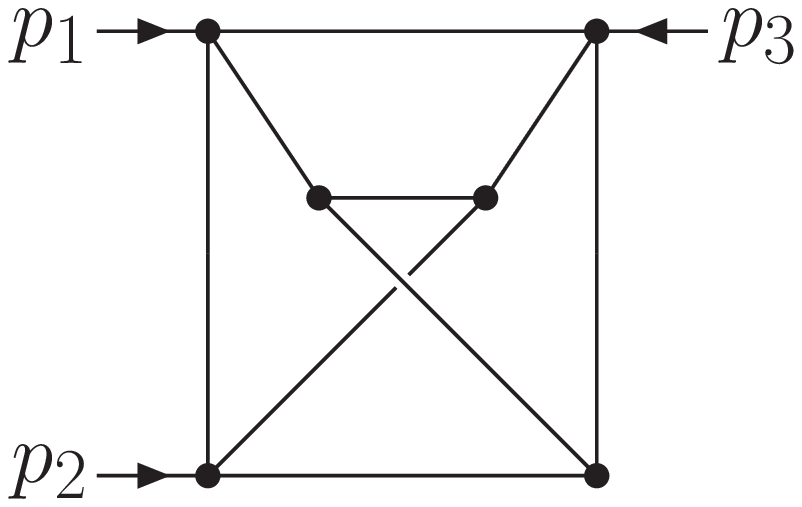}
			\qquad
			\Graph[0.4]{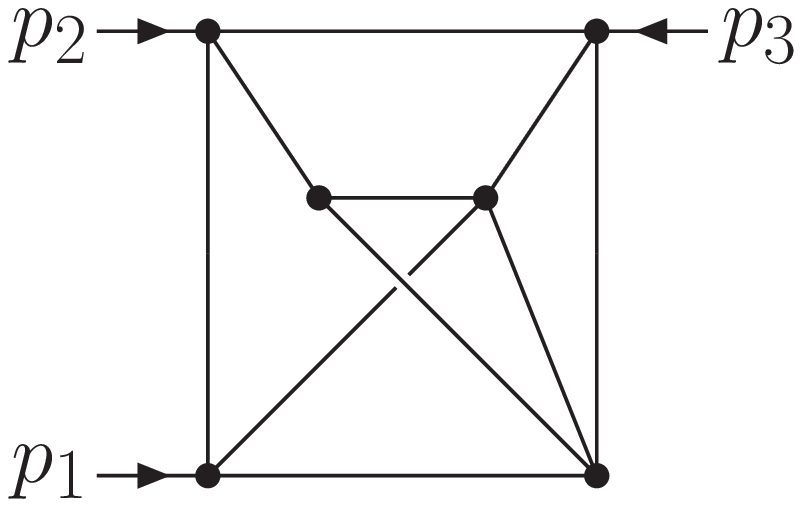}
			\qquad
			\Graph[0.4]{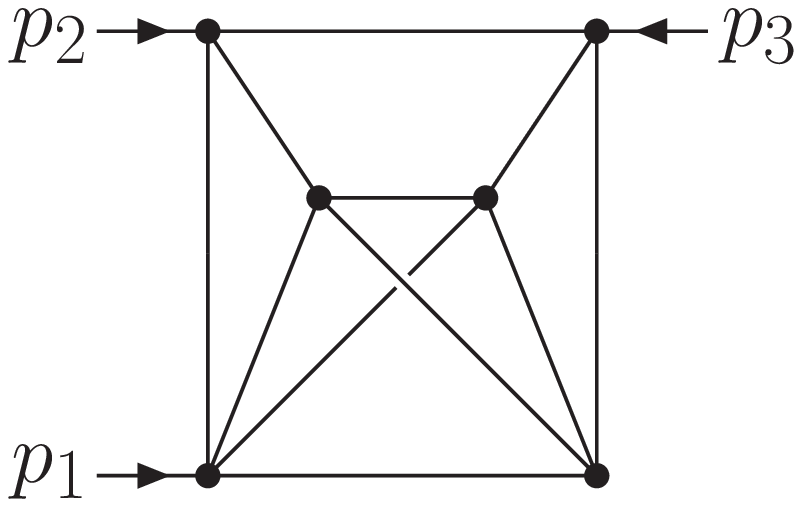}
		\end{equation*}
		\caption{Examples of linearly reducible massless vertices with more than three loops.}
		\label{fig:vertices-moreloops}
	\end{figure}

\subsection{Conformal four-point integrals}%
\label{sec:conformal-four-point}
	The same type of functions that describe off-shell three-point graphs was studied as \emph{graphical functions} in \cite{Schnetz:GraphicalFunctions} and occurs in conformally invariant four-point position-space integrals in exactly $\Dim=4$ dimensions, see \cite{DDEHPS:LeadingSingularitiesOffShellConformal} and references therein.
	Namely, conformal invariance implies that functions like the \emph{hard integral}\footnote{%
	This is introduced in \cite{DDEHPS:LeadingSingularitiesOffShellConformal};
	$x_{ij} \defas \abs{ x_i-x_j}$ denotes Euclidean distances between vectors $x_i,x_j \in \R^4$.%
}
(dashed edges encode propagators in the numerator)
	\begin{equation*}
		H_{12;34}
		= \Graph[0.5]{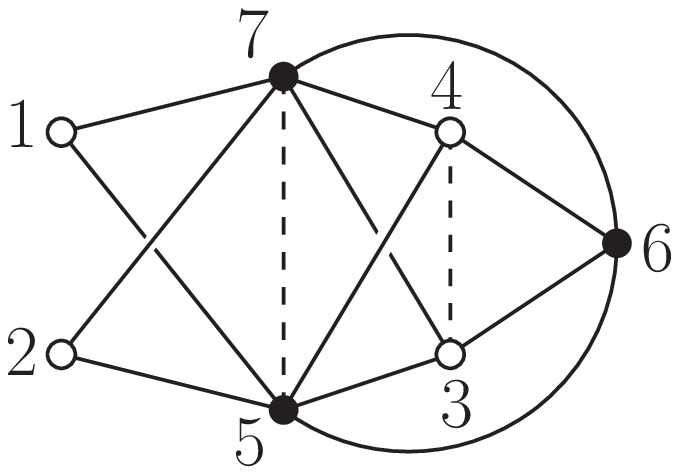}\ 
		\defas 
		\frac{ x_{34}^2 }{\pi^6}\!
		\int\limits_{\R^{12}} \frac{\dd[4]x_5 \, \dd[4] x_6 \, \dd[4]x_7 \cdot x^2_{57}}{(x_{15}^2 x_{25}^2 x_{35}^2
		x^2_{45}) x_{56}^2 (x_{36}^2 x^2_{46}) x^2_{67} (x_{17}^2 x_{27}^2 x^2_{37}
		x_{47}^2)}
	\end{equation*}
	are a product of a rational prefactor and a function depending only on two conformal cross-ratios which can be parametrized in terms of auxiliary variables $z, \bar{z}$ as
	\begin{equation}\label{eq:conformal-cross-ratios}
		z\bar{z} = \frac{x_{12}^2 x_{34}^2}{x_{13}^2 x_{24}^2}
		\quad\text{and}\quad
		(1-z)(1-\bar{z}) = \frac{x_{14}^2 x_{23}^2}{x_{13}^2 x_{24}^2}.
	\end{equation}
	The Schwinger trick delivers a parametric representation for this type of integrals and we found linear reducibility for all such functions at three loops\footnote{In this position-space setting one counts the number of internal vertices as ``loops'' because these are integrated over.} we considered, for example we integrated $H_{12;34}$ and verified the result that was given in \cite{DDEHPS:LeadingSingularitiesOffShellConformal}.
	\begin{figure}
		\begin{gather*}
			\Graph[0.4]{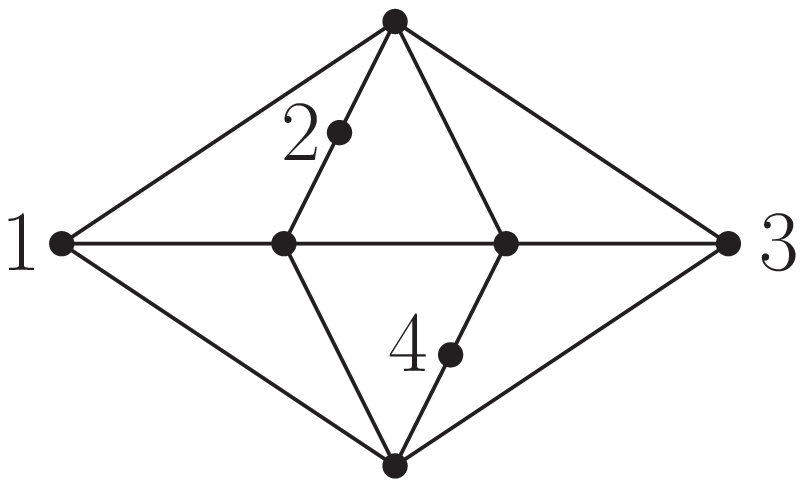}
			\qquad
			\Graph[0.4]{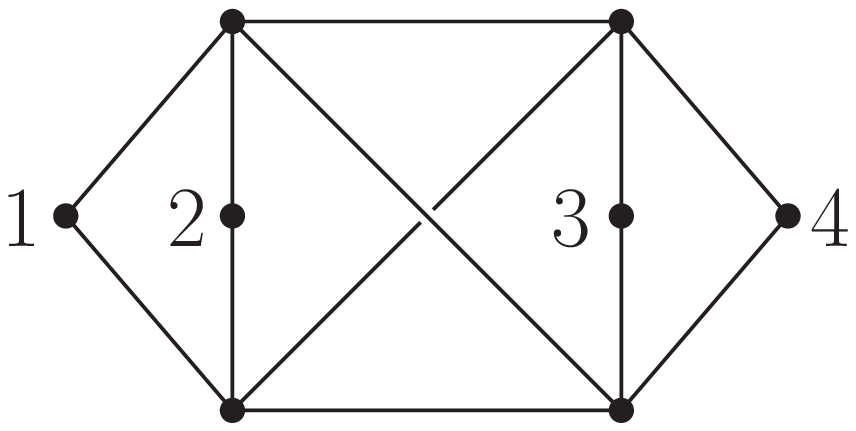}
			\qquad
			\Graph[0.4]{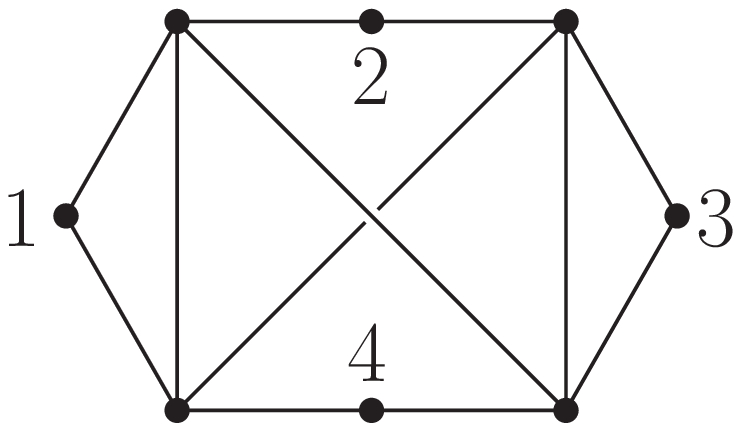}
		\end{gather*}%
		\caption{Conformal four-point graphs (with fixed positions $x_1,\ldots,x_4$ at the vertices marked $1$ through $4$) that are not linearly reducible.}%
		\label{fig:graphical-non-reducible}%
	\end{figure}
	Furthermore, at four loops without inverse (numerator) propagators, all but the three graphical functions in figure \ref{fig:graphical-non-reducible} are linearly reducible and can thus be integrated parametrically.
\begin{example}\label{ex:graphical-reducible}
	The four-point functions depicted in figure \ref{fig:graphical-reducible} are
	\begin{equation}%
		\label{eq:conformal-examples}
		F_{8,10}
		= \frac{f_{8,10}}{x_{34}^2 x_{13}^4 x_{24}^4} 
		\quad
		F_{8,13} 
		= \frac{ f_{8,13}}{x_{34}^2 x_{13}^4 x_{24}^4}
		\quad
		F_{8,16}
		= \frac{f_{8,16}}{x_{13}^4 x_{24}^4} 
	\end{equation}
	for polylogarithms $f_{8,10}, f_{8,13}$ and $f_{8,16}$ provided in the accompanying file. These are of homogeneous weight and feature a common denominator summarized in table \ref{tab:graphical-reducible}.
	The last column counts the summands $L_w(z) \cdot L_u(\bar{z})$ of $f_{8,i}$ with non-zero coefficient in the basis where $u\in \set{0,1}^{\times}$ while $w\in \left(\set{0,1} \cup \Sigma_i \right)^{\times}$ can have additional letters $\Sigma_i \subseteq \set{\bar{z}, \frac{1}{\bar{z}}, 1-\bar{z}}$ given by the zeros of the additional letters of the symbol.

	Rough numeric estimates $f_{8,10} \approx 113$, $f_{8,13} \approx 153$ and $f_{8,16} \approx 552$ at $z=\frac{1}{4}$, $\bar{z}=\frac{1}{2}$ obtained by FIESTA provide a successful check of these exact analytic results ($f_{8,10} = 113.579\ldots$, $f_{8,13} = 154.160\ldots$ and $f_{8,16} = 555.438\ldots$).
\end{example}
\begin{figure}
	\begin{gather*}
		F_{8,10} 
		\defas \Graph[0.4]{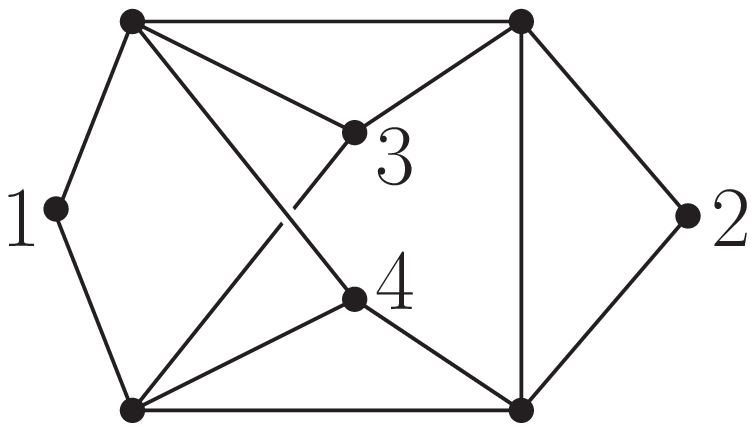}
		\qquad
		F_{8,13}
		\defas \Graph[0.4]{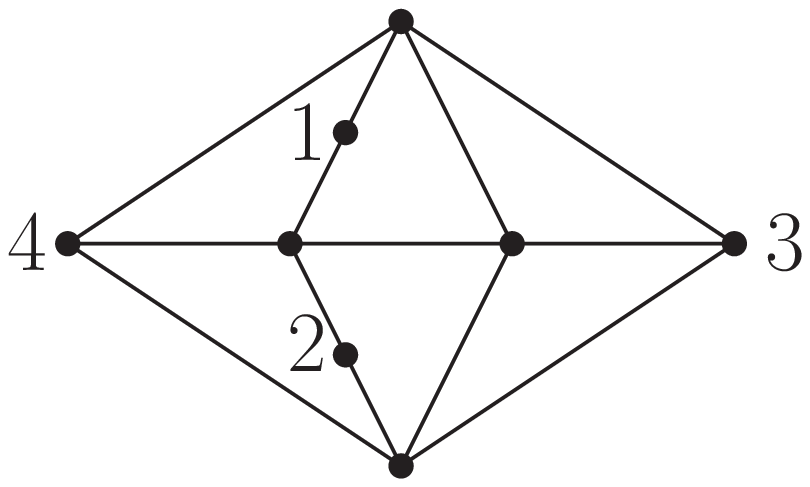}
		\qquad
		F_{8,16}
		\defas \Graph[0.4]{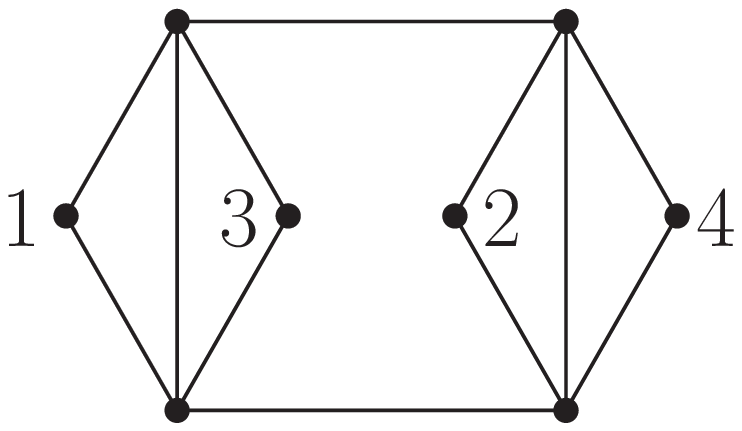}
	\end{gather*}%
	\caption{Linearly reducible graphical functions at four loops of varying complexity from example \ref{ex:graphical-reducible}.}%
	\label{fig:graphical-reducible}%
\end{figure}
\begin{table}
	\centering
	\begin{tabular}{rcccc}
		\toprule
		function		& denominator & weight & additional symbol letters & number of terms \\
		\midrule
		$f_{8,10}$	& $(z-\bar{z})(z+\bar{z}-2)$	& $8$	& $\set{z-\bar{z}, z\bar{z} - 1, z + \bar{z} - 1}$ & $4235$ \\
		$f_{8,13}$	& $(z-\bar{z})^2$ & $8$	& $\emptyset$ & $107$ \\
		$f_{8,16}$	& $z\bar{z}(z-\bar{z})$ & $7$	& $\set{z-\bar{z}}$ & $146$ \\
		\bottomrule
	\end{tabular}
	\caption{Details on the conformal integrals of example \ref{ex:graphical-reducible} (figure \ref{fig:graphical-reducible}).}%
	\label{tab:graphical-reducible}%
\end{table}

\subsection{Integrals with massive propagators and up to seven scales}%
\label{sec:masses-and-many-scales}
	Recently, the method of differential equations was employed to obtain analytic results in terms of polylogarithms for a variety of two-loop integrals involving three scales as for example in \cite{ManteuffelStuderus:MassiveDoubleBoxes,GehrmannTancrediWeihs:TwoLoopqqVVplanar,HennSmirnov:Bhabha}. 
	Clearly it is an interesting question to investigate whether these are linearly reducible; violations of this criterion mean that parametric integration is not possible straight away and might therefore yield to insights how to extend the method as we comment on in section \ref{sec:extending-reducibility}.
	\begin{figure}
		\begin{align*}
			\Graph[0.4]{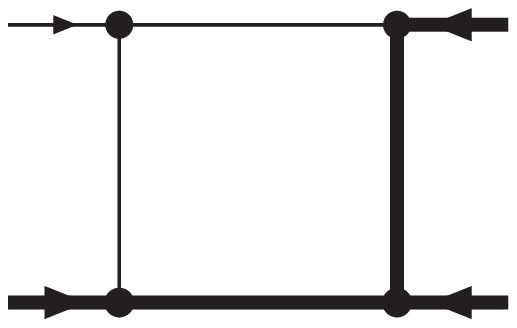}
			\qquad
			\Graph[0.4]{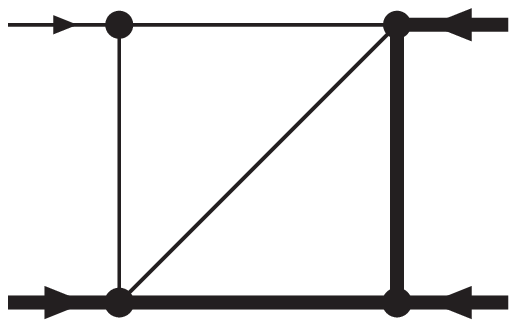}
			\qquad
			\Graph[0.4]{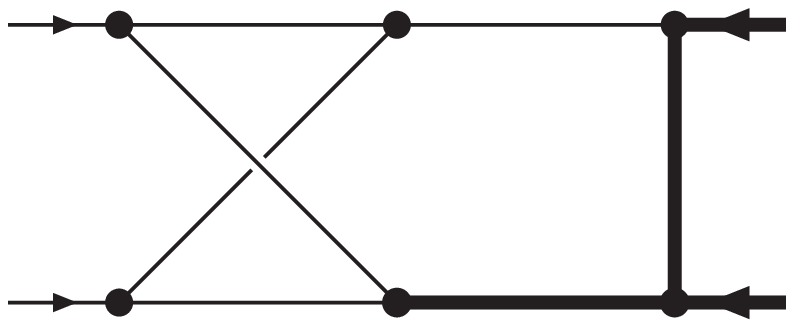}
			\qquad
			\Graph[0.4]{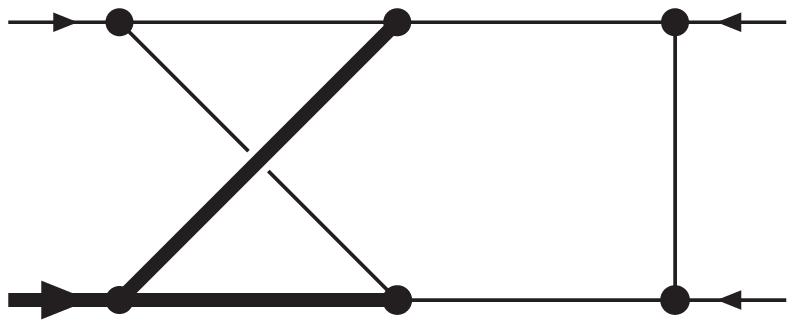}
			\\
			\Graph[0.35]{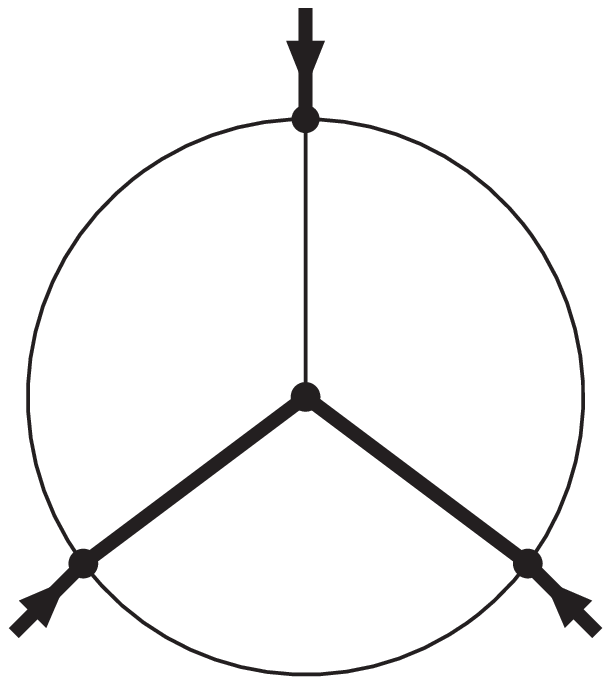}
			\qquad
			\Graph[0.35]{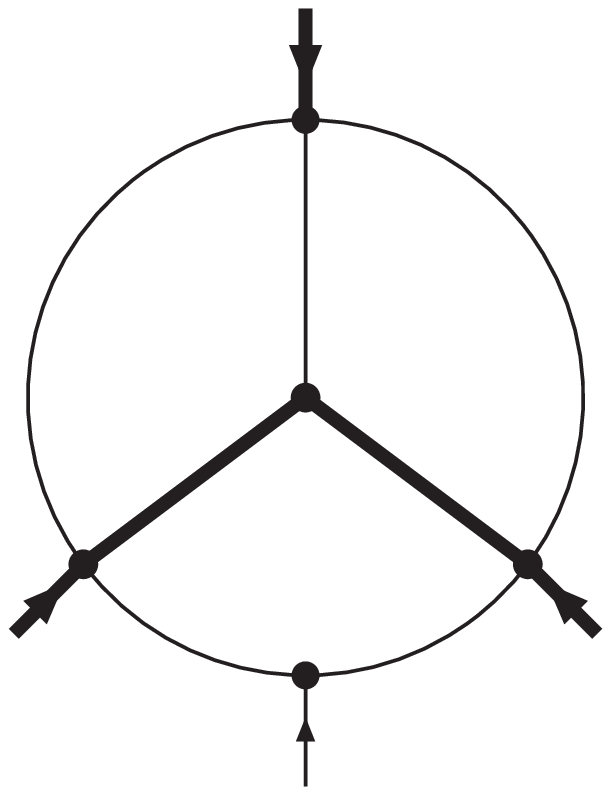}
			\qquad
			\Graph[0.35]{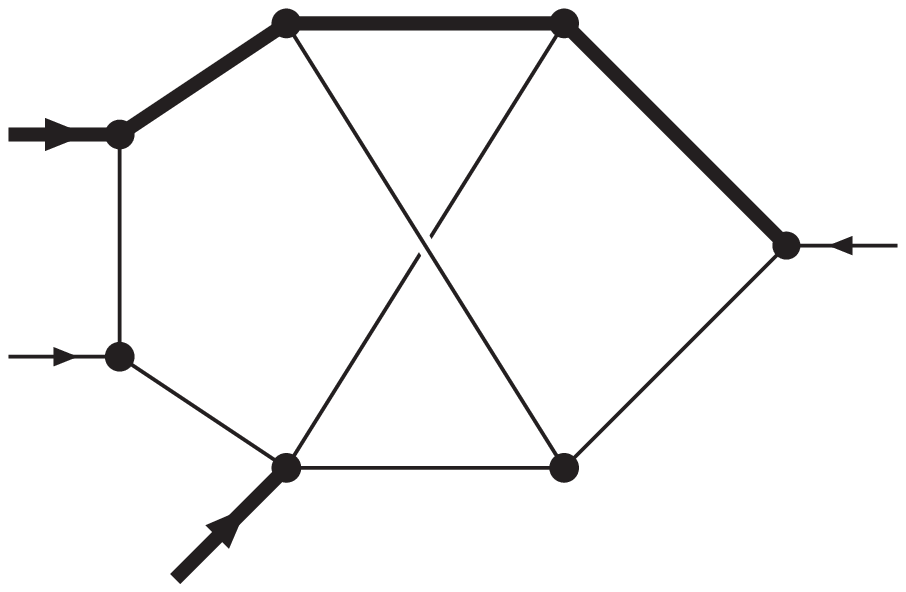}
			\qquad
			\Graph[0.35]{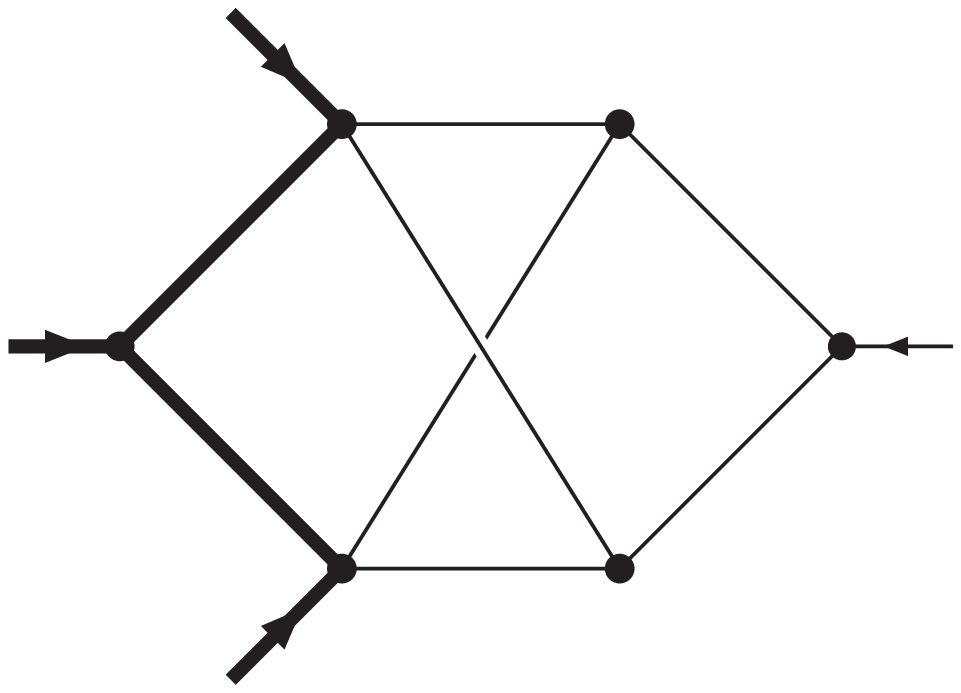}
	\end{align*}
		\caption{Linearly reducible graphs with some massive (thick edges) and otherwise massless propagators (thin edges). Thin external legs are light-like ($p_i^2=0$), while thick external legs may take arbitrary values of $p_i^2$.}%
		\label{fig:reduciblemanyscales}%
	\end{figure}

	Thinking in the other direction, even though most Feynman graphs with general kinematics are \emph{not} linearly reducible, figure \ref{fig:reduciblemanyscales} shows some highly non-trivial integrals we found that are linearly reducible and thus amenable to direct integration. These involve up to three off-shell external momenta and an example with three (different) internal masses.

	As a proof of concept we give explicit results for the first two graphs of figure \ref{fig:reduciblemanyscales} valid in $\Dim=4-2\varepsilon$ dimensions with propagator powers $\ep_e = 1$ for all edges $e$ and Euclidean scalar products $p^2 \geq 0$ of momenta.

\subsubsection{Box with two masses and three off-shell legs (seven scales)}
	The one-loop box with four external momenta and $p_2^2 = m_1 = m_2 = 0$,
	\begin{equation}\label{eq:box2masses}
		\Phi\left( \Graph[0.35]{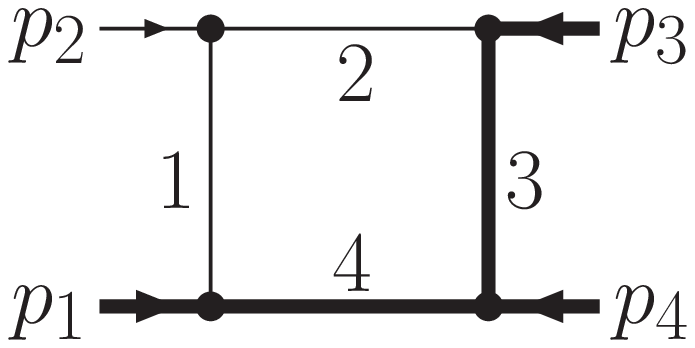} \right)
		= \frac{\Gamma(2+\varepsilon)}{m_4^{4+2\varepsilon}}
		\cdot \sum_{n=-1}^{\infty} f_n \left( \frac{m_3^2}{m_4^2}, \frac{p_1^2}{m_4^2}, \frac{p_3^2}{m_4^2}, \frac{p_4^2}{m_4^2}, \frac{(p_1+p_2)^2}{m_4^2}, \frac{(p_1+p_4)^2}{m_4^2} \right) \varepsilon^n
	\end{equation}
	is linearly reducible (the first graph in figure \ref{fig:reduciblemanyscales}) and can therefore be integrated in Schwinger parameters.
	The arguments of the polylogarithms $f_n$ in general involve several square-roots of rational functions of the six dimensionless ratios, which can be rationalized by quadratic transformations similar to \eqref{eq:triangle-kinematics}.
	For brevity we thus specialize to simpler kinematics in the sequel.

\subsubsection{Box with two adjacent masses and one off-shell leg (five scales)}
\label{sec:boxcorner}%
	Restricting to $p_3^2=p_4^2 = 0$, define the dimensionless ratios
	\begin{equation}\label{eq:boxcorner-kinematics}
		p\defas \frac{p_1^2}{m_4^2},
		\quad
		m\defas \frac{m_3^2}{m_4^2},
		\quad
		s\defas\frac{(p_1+p_2)^2}{m_4^2}
		\quad\text{and}\quad
		u\defas\frac{(p_1+p_4)^2}{m_4^2}
	\end{equation}
	and extract the dependence on $m_4^2$ by power counting such that
	\begin{equation}\label{eq:boxcorner}
		\Phi\left( \Graph[0.35]{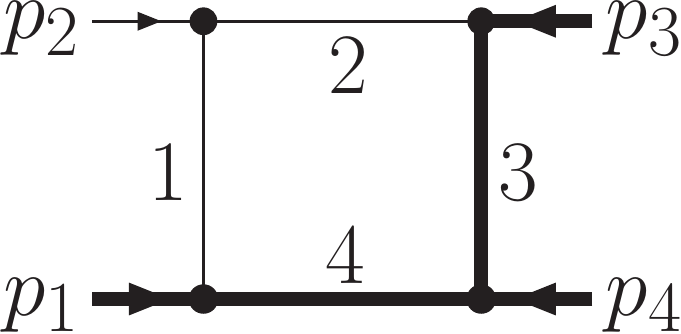} \right)
		= \frac{\Gamma(2+\varepsilon) \cdot m_4^{-4-2\varepsilon}}{m p-s u-s-m u}
			\cdot \sum_{n=-1}^{\infty} f_n \left(s, u, p, m \right) \varepsilon^n.
	\end{equation}
	The final set of polynomials in the reduction (after integrating $\SP_1, \SP_2$ and $\SP_4$) is
	\begin{equation}\begin{split}
		S_{\set{1,2,4}}
		=\Big\{
		&
					s-p,
					s+m, 
					s-m p, 
					s+u-p, 
					s+m-p-1, 
					s(1+u)+m(u-p),
		\\&
					u-p,
					u+1, 
					u+1-m, 
					p+1, 
					1-m%
		\Big\}%
		\end{split}%
		\label{eq:boxcorner-symbol}%
	\end{equation}
	and confines the symbol of the polylogarithms $f_n$ to arbitrary order $n\geq -1$ as explained in section \ref{}.
	In terms of the hyperlogarithms $S_{w} \defas \L_{w}(s)$, $U_{w} \defas \L_{w}(u)$, $P_{w} \defas \L_{w}(p)$ and $M_{w} \defas \L_{w}(m)$ we obtain
	\begin{align}
		f_{-1} &=
U_{-1}
 + S_{-m}
 - P_{-1}
 =
 \ln \frac{(u+1)(s+m)}{m(p+1)}
,
		\\
		f_0 &=
2S_{0,-m}
 - 2S_{{\frac {m \left( -u+p \right) }{u+1}},-m}
 - U_{-1}
 - S_{-m}M_{0}
 - 2S_{-m,-m}
 + 2S_{{\frac {m \left( -u+p \right) }{u+1}}}P_{-1}
\nonumber\\&\quad
 - 2U_{-1,-1}
 - S_{p+1-m}P_{-1}
 + P_{-1}
 + U_{0,-1}
 + U_{-1+m,-1}
 + P_{-1+m}M_{0}
 - P_{0,-1}
\nonumber\\&\quad
 + S_{mp,-m}
 - U_{-1+m}M_{0}
 - 2S_{{\frac {m \left( -u+p \right) }{u+1}}}U_{-1}
 + S_{p+1-m}M_{0}
 - P_{-1+m,-1}
 - S_{-m}
\nonumber\\&\quad
 + 2P_{-1,-1}
 - S_{mp}P_{-1}
 + S_{p+1-m,-m}
 
	\end{align}
	while $f_1$, $f_2$, $f_3$ and $f_4$ are supplied in the attached file. Note that $f_{-1}$ and $f_0$ are given in (4.39) of \cite{EllisZanderighi:ScalarOneLoopQCD} which serve a successful check of our method. We also computed the special case $p_1^2 = -m_4^2$  and the setup $p_1^2=p_2^2=m_4^2=0$ (both introduce a further divergence and thus start proportional to $\varepsilon^{-2}$) to check (4.28) and (4.36) therein.

	The possibility to expand all these integrals to arbitrary order in $\varepsilon$ (further allowing for shifts $\ep_e = 1 + \epe_e \varepsilon$ of propagator powers) is to our knowledge new\footnote{General results in terms of hypergeometric functions are given in \cite{FleischerJegerlehnerTarasov:HypergeometricOneLoop}, however it is not clear how to expand these to arbitrary orders.}.

\subsubsection{Double-triangle with two legs off-shell (four scales)}%
\label{sec:double-triangle-massless}
Consider the second graph of figure \ref{fig:reduciblemanyscales} with massless propagators and two off-shell momenta $q\defas p_3^2$, $p\defas \frac{p_1^2}{q}$ and set $s=\frac{(p_1+p_2)^2}{q}$ and $u=\frac{(p_1+p_4)^2}{q}$. It is linearly reducible along the sequence $3,4,5,2,1$ of edges with final polynomials 
\begin{equation}\label{eq:diagbox-massless-symbol}
	S_{\set{3,4,5,2}}
	=
	\set{p-s, 1-s, 1-u, p - u, p-u s, 1 + p - s - u}
\end{equation}
	which determine the alphabet of the symbol of the functions $f_n$ in the expansion
	\begin{equation}\label{eq:diagbox-massless}
		\Phi\left( \Graph[0.4]{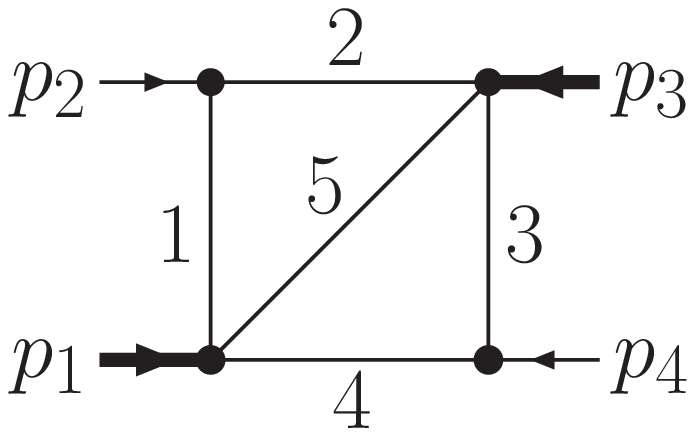} \right)
		=
		\frac{\Gamma(1+2\varepsilon) q^{-1-2\varepsilon}}{1+p-s-u}
		\sum_{n=-2}^{\infty} f_n(s,u,p) \varepsilon^n
		.
	\end{equation}
	Explicitly, in terms of the hyperlogarithms $S_w \defas \L_w(s)$, $U_w \defas L_w(u)$ and $P_w \defas L_w(p)$
	\begin{align}
f_{-2} &=
\frac{\pi^{2}}{2}
 - S_{1,0}
 + S_{{\frac {p}{u}},0}
 - S_{p,0}
 + S_{{\frac {p}{u}}}U_{0}
 - U_{p,0}
 - U_{1,0}
 + P_{0} \left(
 		S_{p}
		+ U_{p}
		- S_{{\frac {p}{u}}}
	\right)
 + P_{0,0}
\label{eq:diagbox-massless-2},
\\
f_{-1} &=
2\mzv{3}
 + {\pi }^{2}
\left(
S_{1-u+p}
 - S_{{\frac {p}{u}}}
 - P_{0}
 + U_{1+p}
 + P_{-1}
\right)
 - 2S_{{\frac {p}{u}}}U_{p}P_{0}
 + 2S_{1-u+p}U_{p}P_{0}
\nonumber\\&\quad
 + 2U_{1,0,0}
 - 4P_{0,0,0}
 + 2P_{-1,0,0}
 + 2U_{p,0,0}
 - 2U_{1+p,1,0}
 - 2S_{1-u+p,1,0}
 - 2S_{1-u+p,p,0}
\nonumber\\&\quad
 + 2S_{{\frac {p}{u}},1,0}
 + 2S_{1,0,0}
 + 2S_{p,0,0}
 - 2U_{1+p,p,0}
 - 2S_{{\frac {p}{u}},0,0}
 + 2S_{1-u+p,{\frac {p}{u}},0}
 + 2S_{{\frac {p}{u}},p,0}
\nonumber\\&\quad
 - 2S_{{\frac {p}{u}},{\frac {p}{u}},0}
 + 2S_{1-u+p,p}P_{0}
 - 2S_{1-u+p,{\frac {p}{u}}}P_{0}
 + 2S_{1-u+p,{\frac {p}{u}}}U_{0}
 + 2S_{1-u+p}P_{0,0}
\nonumber\\&\quad
 - 2S_{p}P_{0,0}
 - 2S_{1-u+p}U_{p,0}
 + 2S_{{\frac {p}{u}}}U_{p,0}
 - 2U_{p}P_{0,0}
 + 2U_{1+p}P_{0,0}
 + 2U_{1+p,p}P_{0}
\nonumber\\&\quad
 + 2S_{{\frac {p}{u}}}U_{1,0}
 - 2S_{1-u+p}U_{1,0}
 - 2S_{{\frac {p}{u}},p}P_{0}
 - 2S_{{\frac {p}{u}},{\frac {p}{u}}}U_{0}
 + 2S_{{\frac {p}{u}},{\frac {p}{u}}}P_{0}
 - 2S_{{\frac {p}{u}}}U_{0,0}
\label{eq:diagbox-massless-1}

	\end{align}
	and we supply $f_0, f_1$ and $f_2$ in the attached file.
	Since $\phipol = q \SP_5 \left( p \SP_1 \SP_4 + \SP_2\SP_3 + s \SP_2\SP_4 + u \SP_1 \SP_3 \right)$ factorizes, we can in fact perform three integrations of \eqref{eq:diagbox-massless} in terms of $\Gamma$-functions and therefore obtain the two-dimensional integral representation
	\begin{equation}\begin{split}
	\Phi\left(  \Graph[0.39]{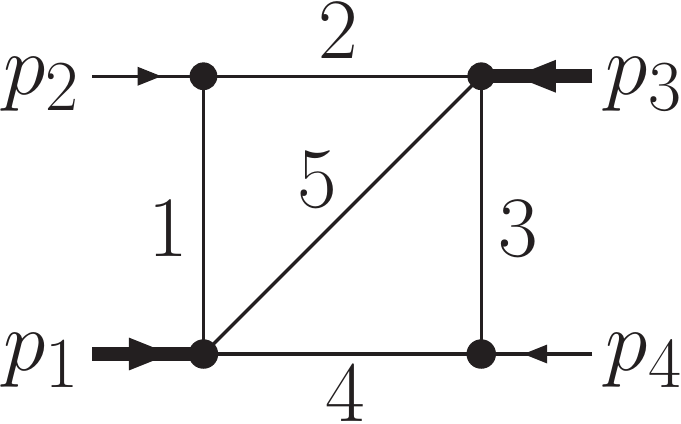} \right)
	&= \frac{
			\Gamma(\sdd)
			\Gamma(\Dim/2-\ep_5)
			\Gamma(\Dim/2-\ep_{12})
			\Gamma(\Dim/2-\ep_{34})
		}{
			\Gamma(\Dim/2-\sdd)
			\Gamma(\ep_1)
			\Gamma(\ep_2)
			\Gamma(\ep_3)
			\Gamma(\ep_4)
			\Gamma(\ep_5)
			\cdot
			q^{\sdd}
		}
		\\& \times
		\int_0^{\infty} x^{\ep_2 - 1} \dd x
		\int_0^{\infty} y^{\ep_4 - 1} \dd y
		\frac{
				(1+x)^{\sdd-\ep_{12}}
				(1+y)^{\sdd - \ep_{34}}
		}{
		\left( u + x + py + s x y \right)^{\sdd}
		}%
	\end{split}%
	\label{eq:double-triangle-2dim-representation}%
\end{equation}
which can be immediately expanded in $\varepsilon$ (linear reducibility is now obvious). We used this second representation to check the results obtained with the (more demanding) five-dimensional integration \eqref{eq:diagbox-massless} and also checked the special case $p=1$ ($p_1^2 = p_3^2$) obtained for $f_{-2}$ and $f_{-1}$ in \cite{GehrmannTancrediWeihs:TwoLoopqqVVplanar} as $\mathcal{I}_{182}^{(B)}$.

\subsubsection{Double-triangle with two legs off-shell and two masses (six scales)}
We now consider the same two-loop graph, but introduce two non-zero masses at the edges 3 and 4. This removes a sub divergence such that the expansion
\begin{equation}%
	\label{eq:diagbox-2mass2off-expansion}%
	\Phi\left( \Graph[0.4]{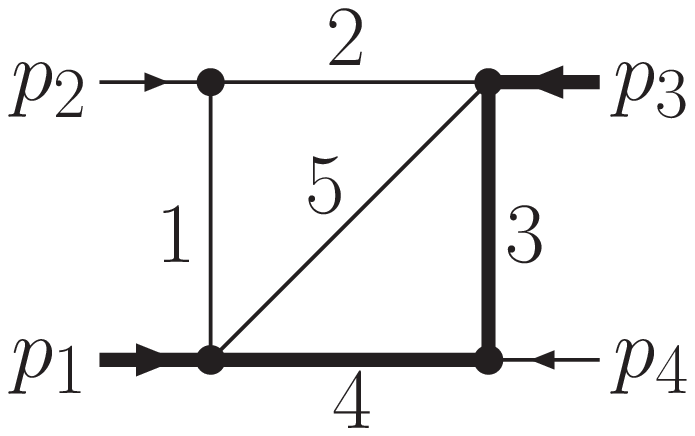} \right)
	= \frac{\Gamma(1+2\varepsilon)}{(p+q-s-u)m_3^{2+4\varepsilon}}
		\sum_{n=-1}^{\infty} f_n(p,s,u,q,m) \cdot \varepsilon^n
\end{equation}
begins with $\varepsilon^{-1}$. In terms of the dimensionless variables
\begin{equation}%
	\label{eq:diagbox-2mass2off-variables}%
	s \defas \frac{(p_1+p_2)^2}{m_3^2}
	\quad
	u \defas \frac{(p_1+p_4)^2}{m_3^2}
	\quad
	p \defas \frac{p_1^2}{m_3^2}
	\quad
	q \defas \frac{p_3^2}{m_3^2}
	\quad
	m \defas \frac{m_4^2}{m_3^2},
\end{equation}
the symbols of all $f_n$ take letters in $\set{s,u,p,q,m} \cup S_{\set{3,4,5,2}}$ for the final polynomials
\begin{align}\label{eq:diagbox-2mass2off-symbol}
	S_{\set{3,4,5,2}}
	=
	\Big\{
	&
		1-m,
		p+m,
		p-s,
		p-u,
		1+q,
		q-s,
		s+m,
		q-u,
		1+u,
		pq-us,
		s-qm,
	\nonumber\\&
		p-um,
		1-p-m+u,
		p-s-u+q,
		p-s+qm-um,
		1-s-m+q,
	\nonumber\\&
		s-pq-qm+us,
		p-us-um+pq,
		pq+p-us-s+qm-um
	\Big\}.
\end{align}
Abbreviating $S_w \defas \Hyper_w (s)$, $U_w \defas \Hyper_w (u)$, $M_w \defas \Hyper_w (m)$, $P_w \defas \Hyper_w (p)$ and $Q_w \defas \Hyper_w (q)$ as before, the leading term becomes
\begin{align}%
	f_{-1} &=
M_0 \left(
	Q_{{0,-1+m}}
	-P_{{u}}U_{{-1+m}}
	+S_{{q}}Q_{{-1+m}}
	-U_{{0,-1+m}}
\right)
-S_{{m \left( -u+q \right) ,qm,-m}}
+P_{{{\frac {us}{q}},s}}S_{{-m}}
\nonumber\\&\quad
-S_{{-m}}P_{{s+um-qm,s}}
+P_{{u}}U_{{-1+m,-1}}
-S_{{q}}Q_{{-1+m,-1}}
+P_{{u,-m+u+1}} \left( U_{{-1}}-M_{{0}} \right)
\nonumber\\&\quad
+U_{{-1}} \left(
	P_{{{\frac {us}{q}},um}}
	-P_{{s+um-qm,um}}
	-P_{{u,um}}
\right)
+ Q_{{-1}} \left(
	S_{{m \left( -u+q \right) ,qm}}
	-S_{{0,qm}}
	+S_{{q,qm}}
\right)
\nonumber\\&\quad
+P_{{s+um-qm,s,-m}}
-P_{{{\frac {us}{q}},um,-m}}
+S_{{0,0,-m}}
-P_{{{\frac {us}{q}},s,-m}}
+P_{{{\frac {us}{q}},0,-m}}-S_{{0,{\frac {m \left( -u+q \right) }{u+1}},-m}}
\nonumber\\&\quad
-S_{{q,qm,-m}}
+S_{{q,-m+q+1,-m}}
-P_{{u,-m+u+1,-m}}
+S_{{0,qm,-m}}
+P_{{s+um-qm,um,-m}}
\nonumber\\&\quad
+P_{{{\frac {us}{q}},-{\frac {-us-s-um+qm}{q+1}},-m}}
+U_{{0,-1+m,-1}}
+S_{{q,0,-m}}
-P_{{u,0,-m}}
-P_{{s,0,-m}}
-Q_{{0,-1+m,-1}}
\nonumber\\&\quad
-P_{{s+um-qm,-{\frac {-us-s-um+qm}{q+1}},-m}}
+S_{{m \left( -u+q \right),{\frac {m \left( -u+q \right) }{u+1}},-m}}
-S_{{m \left( -u+q \right) ,0,-m}}
+P_{{u,um,-m}}
\nonumber\\&\quad
+ \left( U_{{-1}}-Q_{{-1}} \right)\left(
	S_{{m \left( -u+q \right) ,{\frac {m \left( -u+q \right) }{u+1}}}} 
	-S_{{0,{\frac {m \left( -u+q \right) }{u+1}}}}
\right)
+S_{{m \left( -u+q \right) }} \left( -U_{{0,-1}}+Q_{{0,-1}} \right)
\nonumber\\&\quad
+\left(S_{{-m}}+U_{{-1}}-Q_{{-1}} \right) \left[
		P_{{s+um-qm,-{\frac {-us-s-um+qm}{q+1}}}} 
		-P_{{{\frac {us}{q}},-{\frac {-us-s-um+qm}{q+1}}}}
	\right]
+P_{{s}}S_{{0,-m}}
\nonumber\\&\quad
+\left( P_{{s+um-qm}} -P_{{{\frac {us}{q}}}} \right) \left[
	S_{{{\frac {m \left( -u+q \right) }{u+1}}}}
	\left( U_{{-1}}-Q_{{-1}} \right)
	-S_{{{\frac {m \left( -u+q \right) }{u+1}},-m}} 
	+S_{{qm}}Q_{{-1}}
	-S_{{qm,-m}}
\right]
\nonumber\\&\quad
+S_{{q,-m+q+1}} \left( -Q_{{-1}}+M_{{0}} \right) 
+P_{{s+um-qm}} 
\left( 
	Q_{{0,-1}}
	-S_{{0,-m}}
	-U_{{0,-1}}
\right) 
	\label{eq:diagbox-2mass2off-1}%
\end{align}
while $f_0$, $f_1$ and $f_2$ are provided in the ancillary file. Their symbols do not involve the letters $\set{pq+qm-us-s, us+um-pq-p}$ and might suggest that these are indeed superfluous and could be removed from \eqref{eq:diagbox-2mass2off-symbol} by an improved reduction algorithm.

A completely independent check of our analytic results is possible by numeric integration as shown in table \ref{tab:diagbox-2mass2off}. The number in the last row counts the polylogarithms that occur in the basis as used in \eqref{eq:diagbox-2mass2off-1}. Furthermore we checked that the on-shell equal mass limit ($p,q\rightarrow - 1$ and $m \rightarrow 1$) of $f_{-1}$ reproduces the result obtained in \cite{ManteuffelStuderus:MassiveDoubleBoxes}, equations $(3.9)$ and $(3.10a)$.
	\begin{table}
		\centering
		\begin{tabular}{rlllll}
			\toprule
			& $f_{-1}$ & $f_{0}$ & $f_{1}$ & $f_{2}$ & $f_{3}$ \\
			\midrule
			\cite{BognerWeinzierl:ResolutionOfSingularities} & $-0.604907$ &  $+ 0.104586$ & $- 1.03958$ & $+ 0.141365$ & $ - 1.26899$ \\
			exact & $-0.604918601$ & $+0.104721339$ & $-1.039167083$ & $+0.142116843$ &	$-1.267745643$ \\
			terms & $ 66 $ & $ 668 $ & $ 4558 $ & $ 26360 $ & $ 139502 $ \\
			\bottomrule
		\end{tabular}
		\caption{Numeric results for \eqref{eq:diagbox-2mass2off-expansion} at $m_3 = 1$, $m = 2$, $u = 0.75$, $q = 0.5$, $s = 0.2$ and $p = 0.1$ from sector decomposition and first digits of our exact analytic result.}%
		\label{tab:diagbox-2mass2off}%
	\end{table}

\section{Extending linear reducibility}
\label{sec:extending-reducibility}%
	We have seen Feynman graphs that are not linearly reducible but still are known to evaluate to polylogarithms. To us this strongly suggests that in these cases, the Schwinger parameters are not optimal and we expect a different parametrization to exist that allows for parametric integration.

	This idea was already mentioned in \cite{Brown:TwoPoint} and we like to demonstrate how a rational parametrization of quadrics can indeed restore linear reducibility (in a different set of variables). 
	In principle, this technique can always be applied if the obstruction to linear reducibility is given by a single quadratic polynomial.

	\subsection{One-loop example: box with two masses vis-\`{a}-vis}
	Consider the on-shell massive box with $p_1^2=p_2^2=p_3^2=p_4^2 = -m^2 $ for two massive propagators $m_1=m_3=m$ and massless $m_2=m_4=0$ as shown in figure \ref{fig:boxopposing}.
	\begin{figure}
		\begin{equation*}
			\Graph[0.5]{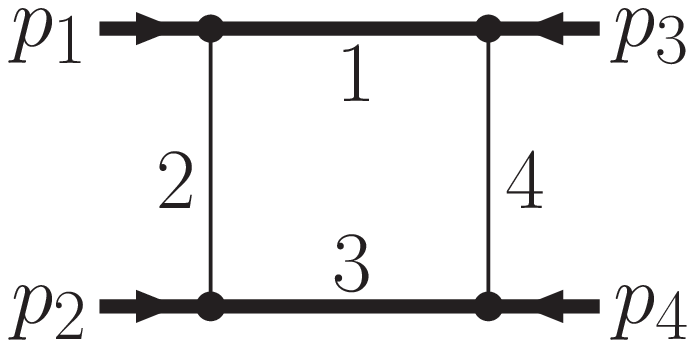}
			\qquad
			\Graph[0.4]{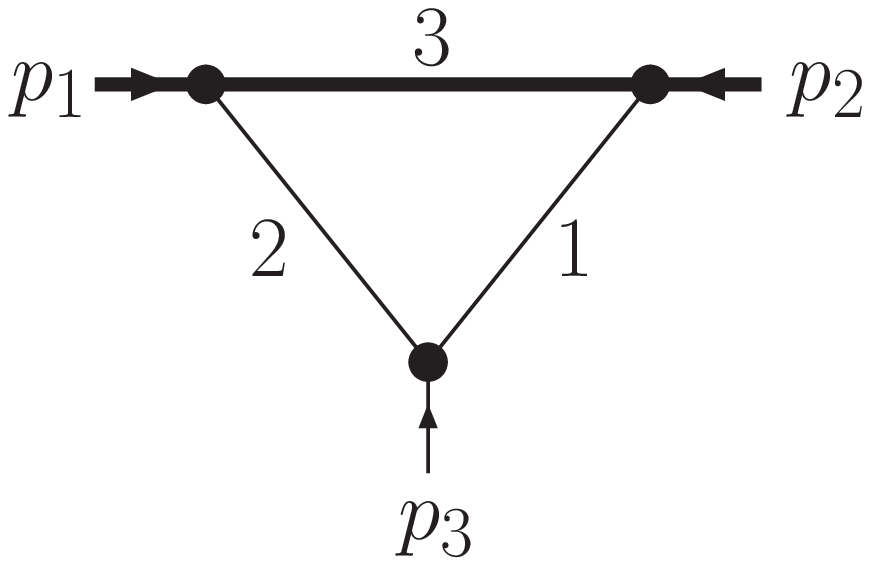}
		\end{equation*}
		\caption{A box with two massive internal lines ($1$ and $3$) that are not adjacent is not linearly reducible. The infrared divergence of the (linearly reducible) triangle graph is studied in example \ref{ex:triangle-divergences}.}%
		\label{fig:boxopposing}%
	\end{figure}
	In contrast to the case of section \ref{sec:boxcorner} where the massive propagators are adjacent, this graph is not linearly reducible: With $s=(p_1+p_2)^2$ and $t=(p_1+p_3)^2$, its graph polynomials are
	\begin{equation}\label{eq:boxopposite-polynomials}
		\psipol
		= \SP_1 + \SP_2 + \SP_3 + \SP_4
		\quad\text{and}\quad
		\phipol
		= s \SP_1 \SP_3 + t\SP_2\SP_4 + m^2 \left( \SP_1 + \SP_3 \right)^2
	\end{equation}
	and $\phipol$ is only linear in $\SP_2$ and $\SP_4$. Reducing (integrating) $\SP_4$ we obtain the set
	\begin{align}
		S_{\set{4}}
		=& \set{\SP_1+\SP_2+\SP_3,
					s \SP_1 \SP_3 + m^2 \left( \SP_1 + \SP_3 \right)^2,
					R}
		\quad\text{where the resultant}%
		\label{eq:boxopposite-reduction-4}\\
		R 
		\defas&
		\resultant[\SP_4]{\psipol}{\phipol}
		=
		s \SP_1 \SP_3 + m^2 \left( \SP_1 + \SP_3 \right)^2 -t\SP_2(\SP_1+\SP_2+\SP_3)
		\label{eq:boxopposite-resultant}%
	\end{align}
	is irreducible and quadratic in all remaining Schwinger parameters, therefore prohibiting any further integration.
	To proceed we change variables according to
	\begin{equation}
		\frac{s}{m^2}
		=\frac{(1-x)^2}{x}
		\quad
		\frac{t}{m^2}
		= \frac{(1-y)^2}{y}
		\quad
		\SP_2 = \widetilde{\SP_2} \left( \SP_1 + x\SP_3 \right)
		\quad
		\SP_4 = \widetilde{\SP_4} \left( x\SP_1 + \SP_3 \right).
		\label{eq:boxopposite-variable-change}%
	\end{equation}
	On one hand we reparametrized the kinematics via $x$ and $y$ to rationalize roots that would otherwise appear in the result (this is analogous to \eqref{eq:triangle-kinematics}), while afterwards
	\begin{equation}
		\phipol
		= \frac{m^2}{x} ( \SP_1 + x \SP_3 ) ( x \SP_1 + \SP_3 )
			+\frac{m^2}{y} (1 -y)^2 \SP_2 \SP_4		
		\label{}
	\end{equation}
	suggests to introduce the variables $\widetilde{\SP_2}$ and $\widetilde{\SP_4}$ of \eqref{eq:boxopposite-variable-change} with the effect that
	\begin{equation}
		\phipol = \frac{m^2}{xy}
							(\SP_1 + x \SP_3 ) (x \SP_1 + \SP_3)
							\left[ y+ x(1-y)^{2} \widetilde{\SP_2} \widetilde{\SP_4} \right]
		\label{eq:boxopposite-phi-new-variables}%
	\end{equation}
	factors linearly in these new parameters. It follows that $\widetilde{R} = \resultant[\widetilde{\SP_4}]{\psipol}{y + x(1-y)^2 \widetilde{\SP_2}\widetilde{\SP_4}}$ is linear in $\SP_1$ and $\SP_3$ allowing for a further integration. Calculating the reduction shows that we can finally also integrate $\widetilde{\SP_2}$ and obtain the final set of polynomials
	\begin{equation}
		S_{\set{\widetilde{\SP_2},\SP_3,\widetilde{\SP_4}}}
		= \set{x+1,x-1,y+1,y-1,xy+1,x+y}
		\label{eq:boxopposite-symbol-alphabet}%
	\end{equation}
	which together with $\set{x,y}$ define the alphabet of the symbol of the resulting function of $x$ and $y$.
	This coincides with the observation made in \cite{HennSmirnov:Bhabha} upon a study of its differential equations.
	With the described change of variables we applied the parametric integration procedure and cross-checked our result successfully with the expansion given as (2.27) in \cite{HennSmirnov:Bhabha}.

	\subsection{Three-loop example: $K_4$}\label{sec:K4-variable-change}
	We return to the complete graph $K_4$ of figure \ref{fig:moreloop-on-shell} with massless on-shell kinematics $p_1^2=\ldots=p_4^2=m_1=\ldots=m_6=0$ already mentioned in section \ref{sec:massless-onshell-4pt}.
	In this case, after integrating say $\SP_2$ we can not proceed further because again the resultant $R \defas \resultant[2]{\psipol}{\phipol} \in S_{\set{2}}$ is irreducible and quadratic in all Schwinger parameters. But its discriminant\footnote{For $R=A \SP_3^2 + B \SP_3 + C$ the discriminant with respect to $\SP_3$ is $\discriminant[\SP_3]{R} = B^2 - 4 AC$.}
	\begin{align}\label{eq:K4-discrim}%
		\discriminant[\SP_3]{R}
		&=
\SP_{1}^{2}
\left( \SP_{1}\SP_{6}+\SP_{6}\SP_{4}+\SP_{4}\SP_{5}+\SP_{6}\SP_{5} \right)^{2}
\\ & \qquad \times
\left[ \SP_{6}^{2} {t}^{2} \SP_{4}^{2}-2 t\SP_{4}\SP_{6} \left( -t\SP_{6}+s\SP_{4}-s\SP_{6} \right) \SP_{5}
	+ \left( s\SP_{6}+s\SP_{4}+t\SP_{6} \right)^{2} \SP_{5}^{2} \right] \nonumber
	\end{align}
	becomes a perfect square if we introduce a new variable $\xi$ and reparametrize
	\begin{equation}\label{eq:K4-substitution}
		\SP_{5}
		= \frac{
				t\SP_{4}\SP_{6} \left( s\SP_{6}+t\SP_{6}+\xi \right)
			}{
				\xi \left( s\SP_{6}+s\SP_{4}+t\SP_{6}+\xi \right) 
			}.
	\end{equation}
	Hence after this transformation, $R$ factorizes linearly in $\SP_3$ and indeed the polynomial reduction shows that we obtain linear reducibility along the sequence $\SP_2,\SP_3,\SP_1,\SP_4,\xi$ ($\SP_6=1$) of integrations. The final set is $\set{s+t}$ and proves that to all orders, $\Phi\left( K_4 \right) s^{3\varepsilon}$ is a harmonic polylogarithm of $\frac{t}{s}$ which was observed before in \cite{HennSmirnov:K4}. We performed the explicit integrations and reproduced the result up to order $\varepsilon^2$ (polylogarithms of weight six) given in (B.1) of \cite{HennSmirnov:K4}.

\section{Divergences in Schwinger parameters}
\label{sec:dimreg}%
The method of parametric integration relies on convergent integral representations of the quantities (functions) to be computed, but many Feynman integrals are divergent.
While ultraviolet divergences can be renormalized on the level of the integrand directly\footnote{Most standard text books on quantum field theory explain the BPHZ method, e.g. \cite{Collins,ItzyksonZuber}.} and then result in a convergent parametric integral representation (see for example \cite{BrownKreimer:AnglesScales}), the cancellation of infrared divergences is more subtle.
In practical calculations it turned out to be most useful to assign values also to infinite integrals in terms of a regularization prescription, therefore separating the two problems of calculation of the integrals and renormalization of their divergences.

In this section we briefly explain why the most widely employed \emph{dimensional regularization}\footnote{A definition of this scheme in terms of convergent \emph{momentum space} integrals is given in \cite{Collins}.} is perfectly adapted to parametric integration and explain a general method to generate convergent integral representations of dimensionally regulated, divergent Feynman integrals.
Note that usually this task is solved by the method of \emph{sector decomposition} \cite{BinothHeinrich:SectorDecomposition,BinothHeinrich:SectorDecomposition2} which is publicly available as \cite{BognerWeinzierl:ResolutionOfSingularities,BorowkaCarterHeinrich:SecDec2,SmirnovTentyukov:FIESTA}. But this approach introduces various changes of variables and decomposes the original integrand into many summands, which would need to be analyzed separately for linear reducibility. Furthermore finite integrals are obtained by subtraction of counterterms, and we argued in \cite{Panzer:MasslessPropagators} that it is in general unclear how this effects the polynomial reduction.

Therefore we prefer an expression in the original Schwinger parameters, involving only the polynomials $\psipol$ and $\phipol$ in denominators. The criteria of convergence here are well-known and we merely employ integration by parts, so we do certainly not assume our result to be new but rather a reformulation.
Nevertheless it is crucial for our study of linear reducibility.

We investigate a projective parametric integral $\left[\prod_e \int_0^\infty \dd \SP_e\right] F \delta(H)$ which we denote by $\int F \Omega$ (so $\Omega$ is the canonical volume form on $\RP^{\abs{E}-1}$). The parametric integrand is a rational function of the $\SP_e$ and contains exponents that can depend on $\Dim$ and the propagator powers $\ep_e$. Given disjoint sets $J, K\subset E$ of edges, 
\begin{equation}%
	\label{eq:divergence-rescaling}
	F_J^{K} 
	\defas
	\restrict{F}{
		\SP_e \mapsto \lambda \SP_e\ \forall e\in J
		\ \text{and}\ 
		\SP_e \mapsto \lambda^{-1} \SP_e\ \forall e\in K
	}
	\in \bigo{\lambda^{\ldeg_J^{K}(F)}}
	\quad\text{at}\quad
	\lambda \rightarrow 0
\end{equation}
defines a degree $\ldeg_J^K(F)$ of vanishing\footnote{So $\ldeg_J^K(F)$ is the unique number $s$ such that $\lim_{\lambda\rightarrow 0} \left[ \lambda^{-s} \cdot F_J^K \right]$ exists and is non-zero.} of $F$ when all $\SP_e$ with $e \in J$ tend to zero and $\SP_e\rightarrow \infty$ for $e \in K$.
Denoting the associated degree of divergence by
\begin{equation}
	\SDD{J}{K}(F)
	\defas 
	\abs{J}
	- \abs{K}
	+\ldeg_J^K(F),
\end{equation}
we recall the well-known finiteness result 
\begin{lemma}
	Let all non-zero coefficients of $\phipol$ be positive\footnote{Otherwise divergences can occur inside the integration domain. Some examples of this more complicated situation are explained in \cite{JantzenSmirnov:PotentialAndGlauberRegions}.} and $\SDD{J}{K} > 0$ for all disjoint $J,K \subset E$ with $\emptyset \neq J \cupdot K \subsetneq E$. Then $\int F \Omega$ is absolutely convergent.
	\label{}
\end{lemma}
\begin{proof}
	The positivity condition implies that $F$ is continuous on the interior $(0,\infty)^{\abs{E}}$, hence $F$ can only have singularities on the boundary $\bigcup_e \left( B_e^0 \cupdot B_e^{\infty} \right)$ where $B_e^{\bullet} \defas \set{\SP_e=\bullet}$. Cover the domains $\R_+ = (0,1] \cup [1,\infty)$ and transform $\SP_e\rightarrow \SP_e^{-1}$ on $(0,1]$ such that
	\begin{equation*}
		\int F \Omega
		= \sum_{I \subseteq E}
			\Big[ 
				\prod_{e \notin I}
				\int_1^\infty \dd \SP_e
			\Big]
			\Big[ 
				\prod_{e \in I}
				\int_1^\infty \frac{\dd \widetilde{\SP}_e}{\widetilde{\SP}_e^2}
			\Big]
			\restrict{F}{\SP_e = \widetilde{\SP_e}^{-1}\ \forall e \in I}.
	\end{equation*}
	Now $\SDD{J}{K} > 0$ for all $K \subset E \setminus I$ and $J \subset I$ translates into the convergence condition\footnote{As the projective integral is only $(E-1)$ dimensional, we can drop constraints on $\SDD{J}{K}$ when $J \cup K = E$.} of Weinberg's Theorem \cite{Weinberg:HighEnergy}.
\end{proof}
This result is well-known and a graph-theoretic interpretation of $\SDD{J}{K}$ is possible, see for example the appendix E.1 of \cite{Smirnov:EvaluatingFeynmanIntegrals} and references therein.

\begin{example}\label{ex:triangle-divergences}%
	The triangle graph $G$ of figure \ref{fig:boxopposing} with one internal mass $m=m_3$ and light-like $p_3^2 = 0$ is linearly reducible with the graph polynomials
	\begin{equation}%
		\label{eq:divergent-vertex-polynomials}%
		\psipol
		= \SP_1 + \SP_2 + \SP_3
		\quad\text{and}\quad
		\phipol
		= \SP_3 \left( m^2 \psipol + p_1^2 \SP_2 + p_2^2 \SP_1 \right).
	\end{equation}
	For $\Dim = 4-2\varepsilon$ and $\ep_1=\ep_2=\ep_3=1$, the parametric integral representation
	\begin{align*}
		\Phi(G)
		&= \int \frac{\dd[\Dim] k}{\pi^{\Dim/2}}
			 \frac{1}{(k^2 + m^2)(k+p_2)^{2}(k-p_1)^{2}}
		= \Gamma(1+\varepsilon)
			 \int \frac{\Omega}{\psipol^{1-2\varepsilon} \phipol^{1+\varepsilon}} 
	\end{align*}
	diverges at $\varepsilon = 0$: $\SDD{\set{3}}{\emptyset} = \SDD{\emptyset}{\set{1,2}} = -\varepsilon$ represent a logarithmic divergence. It is apparent from the factor $\SP_3^{-1-\varepsilon}$ in the integrand.
\end{example}

\subsection{Analytic regularization}
For a choice of disjoint $J, K \subset E$ we can regard $\lambda$ as a new integration variable by inserting the factor $1 = \int_0^{\infty} \dd\lambda\ \delta\left(\lambda - \SP_J-\SP_K^{-1} \right)$ where $\SP_J \defas \sum_{e\in J} \SP_e$. After rescaling $\SP_e$ by $\lambda$ ($\lambda^{-1}$) for $e \in J$ ($e \in K$), we see
	\begin{equation*}
		\int F \Omega
		= \int \Omega\ \delta\left( 1-\SP_J - \SP_K^{-1} \right)
		\int_0^{\infty} \frac{\dd \lambda}{\lambda} \lambda^{\SDD{J}{K}} \cdot \widetilde{F_J^K}(\lambda)
	\end{equation*}
	where $\widetilde{F_J^K} \defas F_J^K \cdot \lambda^{-\ldeg_J^K}$ is finite at $\lambda\rightarrow 0$. The partial integration
	\begin{equation}\label{eq:anareg-partial}
		\int_0^{\infty} \frac{\dd \lambda}{\lambda} \lambda^{\SDD{J}{K}} \cdot \widetilde{F_J^K}(\lambda)
		= \restrict{\frac{\lambda^{\SDD{J}{K}}}{\SDD{J}{K}} \widetilde{F_J^K}(\lambda)}{\lambda=0}^{\infty} 
		- \frac{1}{\SDD{J}{K}} \int_0^{\infty} \dd \lambda \cdot \lambda^{\SDD{J}{K}} \frac{\partial}{\partial \lambda} \widetilde{F_J^K}(\lambda)
	\end{equation}
	has vanishing boundary contribution when $\SDD{J}{K}>0$ and $F_J^K(\lambda)$ falls off at $\lambda\rightarrow \infty$ faster than $\lambda^{\abs{J}-\abs{K}}$, so in particular whenever $\int F \Omega$ is convergent.

	The analytically regularized functions associated to both integrals in \eqref{eq:anareg-partial} are therefore equal, because they are meromorphic in the analytic regulators $\set{\Dim} \cup \setexp{\ep_e}{e\in E}$ and coincide in the domain of absolute convergence of $\int F \Omega$ (which is non-empty as proven already in \cite{Speer:SingularityStructureGenericFeynmanAmplitudes}).

	Changing back to the original variables, we conclude that $\int \Omega F = \int \Omega\ \anapartial{J}{K} \left( F \right)$ where
	\begin{equation}\label{eq:anapartial}
			\anapartial{J}{K}
			\defas
			1-\frac{1}{\SDD{J}{K}} \left[
				\sum_{e\in J} \partial_e \SP_e
				-\sum_{e\in K} \partial_e \SP_e
			\right]
			= \frac{1}{\SDD{J}{K}} \left[
					\ldeg_J^K
					-\sum_{e\in J}\SP_e \partial_e
					+\sum_{e\in K}\SP_e \partial_e
			\right]
	\end{equation}
	denotes a differential operator with $\partial_e \defas \frac{\partial}{\partial \SP_e}$.
	\begin{example}[Triangle graph of figure \ref{fig:boxopposing}]
		With respect to $J=\set{3}$ and $K=\emptyset$ we have $\widetilde{F_J^K} = \psipol^{2\varepsilon-1} \cdot \left[ m^2 \psipol + p_2^2\SP_1 + p_1^2\SP_2 \right]^{-1-\varepsilon}$ with $\SDD{J}{K} = -\varepsilon$ and deduce
		\begin{equation}\label{eq:triangle-divergent-partial}
			\int \frac{\Omega}{\psipol^{1-2\varepsilon} \phipol^{1+\varepsilon}}
			= 
				\frac{1}{\varepsilon} \cdot \int \frac{\Omega}{\SP_3^{\varepsilon}} \frac{\partial}{\partial \SP_3} \widetilde{F}
			= \frac{1}{\varepsilon} \cdot 
				\int \frac{\Omega\ \SP_3 }{\psipol^{1-2\varepsilon} \phipol^{1+\varepsilon}}
				\left[ \frac{2\varepsilon-1}{\psipol}
				- \frac{(1+\varepsilon)\SP_3 m^2}{\phipol}
				\right]
		\end{equation}
		as an identity between analytically regularized integrals. In their joint domain $\varepsilon<0$ of convergence, the boundary term $\restrict{\frac{\SP_3^{-\varepsilon} \widetilde{F}}{-\varepsilon}}{\SP_3=0}^{\infty}$ is well-defined and vanishes. Note that the integral on the right-hand-side of \eqref{eq:triangle-divergent-partial} has an increased regime $\varepsilon<1$ of convergence.
	\end{example}
	We can summarize our results in the form of
	\begin{lemma}\label{lemma:anapartial}
		For any disjoint subsets $J, K \subset E$ with $\emptyset \neq J \cup K \subsetneq E$, the new parametric integrand $\widetilde{F} \defas \anapartial{J}{K}\left( F \right)$ fulfils
		\begin{enumerate}
			\item $\int F \Omega$ = $\int \widetilde{F} \Omega$ as analytically regularized integrals,
			\item $\SDD{J'}{K'} \left( \widetilde{F} \right) \geq \SDD{J'}{K'} \left( F \right)$ for any disjoint $J', K' \subset E$ and
			\item $\SDD{J}{K} \left( \widetilde{F} \right) \geq 1 + \SDD{J}{K}\left( F \right)$ increases at least by one.
		\end{enumerate}
	\end{lemma}
	Properties 2 and 3 are probably most evident by introducing simultaneously variables $\lambda_J^K$ for all disjoint $J,K$ and rescaling $\SP_e \mapsto \SP_e \prod_{e \in J, J \cap K = \emptyset} \lambda_J^K \prod_{e \in K, J \cap K = \emptyset} \left( \lambda_J^K \right)^{-1}$ such that
	\begin{equation}\label{eq:anareg-rescalings}
		F
		= \prod_{J \cap K = \emptyset} \left( \lambda_J^K \right)^{\ldeg_J^K}
			\cdot
			R
		\quad\text{where}\quad
		R = \prod_{p \in \mathcal{P}} p^{\ep_p}
	\end{equation}
	factors into irreducible polynomials\footnote{For scalar Feynman integrals we have only the irreducible factors of $\psipol$ and $\phipol$ in $\mathcal{P}$.} $p \in \mathcal{P}$ (with exponents $\ep_p$) in Schwinger parameters $\SP_e$, the scaling variables $\lambda_J^K$ and kinematic invariants. Now the action of $\anapartial{J}{K}$ on $F$ equals replacing $R$ by $\lambda_J^K\partial_{\lambda_J^K} R$ (and dividing by $-\SDD{J}{K}$), but
	\begin{equation*}
		\partial_{\lambda_J^K} R
		= R \sum_{p \in \mathcal{P}} \frac{\ep_p}{p} \cdot \partial_{\lambda_J^K} \big( p \big)
	\end{equation*}
	can only factor in the numerator (then possibly contributing additional powers of some $\lambda_{J'}^{K'}$), while the extra denominators $p$ do by construction not introduce new divergences (which would correspond to poles at $\lambda_{J'}^{K'} \rightarrow 0$).
	\begin{corollary}
		Finitely many applications of operators $\anapartial{J}{K}$ on $F$ suffice to generate a representation of $\int F\Omega = \int \widetilde{F} \Omega$ with a convergent parametric integrand $\widetilde{F}$ (all $\SDD{J}{K}\left(\widetilde{F} \right)$ are positive).
	\end{corollary}
	For example, when $\int F \Omega$ is regulated by $\Dim$ alone\footnote{This is not always possible; in general divergences require the exponents $\ep_e$ as analytic regulators.}, we identify divergences $\restrict{\SDD{J}{K}}{\varepsilon=0} \geq 0$ by power-counting and apply $\anapartial{J}{K}$ sufficiently often until $\restrict{\SDD{J}{K}}{\varepsilon=0} > 0$.

	Crucially, the representation $\widetilde{F}$ obtained this way can only contain $\psipol$ and $\phipol$ with non-integer or negative exponents. Therefore, any term in its $\varepsilon$-expansion lies in
	\begin{equation}
		\Q\left[
				\kinematics
				\cup
				\setexp{\SP_e}{e\in E}
				\cup \set{\psipol^{-1}, \phipol^{-1}}
				\cup \set{\ln \psipol, \ln \phipol}
			\right]
		\label{}
	\end{equation}
	and can be integrated using hyperlogarithms precisely when the graph under consideration is linearly reducible.
	Put differently, the partial integrations $\anapartial{J}{K}$ do not affect the polynomial reduction.

	We applied this technique for all explicit computations of subdivergent integrals in this article, namely example \ref{example:4pt-4loop}, $\Delta_{3,14}$ from example \ref{example:3loop-vertices} and all results of section \ref{sec:masses-and-many-scales}.

\section{Summary and outlook}%
\label{sec:summary}

We extended the method of parametric integration to divergent, analytically regularized Feynman integrals $G$ for linearly reducible $G$ that may depend on multiple kinematic invariants. Several non-trivial examples were shown and explicit new results given in terms of polylogarithms. Let us stress that such a graph $G$ can in principle be computed
	\begin{itemize}
		\item	to arbitrary order in $\varepsilon$, expanded near any even dimension $\restrict{\Dim}{\varepsilon=0} \in 2\N$,
		\item including any tensor structure (loop momenta in the numerator); in particular the form-factor-decomposition is automatic in the parametric representation and we do not need a reduction to master integrals in the integration-by-parts (IBP) sense,
		\item with arbitrary powers $\ep_e = n_e + \varepsilon \epe_e$ of propagators for $n_e \in \Z$.
	\end{itemize}
Practically however, tensor structure and sub divergences (via the integrand preparation of section \ref{sec:dimreg}) can result in very complicated initial integrands, involving high powers of $\psipol$ and/or $\phipol$ in the denominator and a huge polynomial in the numerator. Such cases require a simplification before the computation and it seems possible to apply the idea of integration by parts directly to these parametric integrands which we will try to return to in the future.

In this context also note that the procedure suggested by lemma \ref{lemma:anapartial} seems to generate unnecessarily complicated integrands in the case of overlapping divergences.
\begin{example}
	In the case of the two-loop graph of section \ref{sec:double-triangle-massless}, the original parametric integrand has many divergences. E.g. we find $\SDD{\emptyset}{\set{\SP_1, \SP_2, \SP_3, \SP_4}} = -2 \varepsilon$, $\SDD{\set{\SP_3,\SP_4,\SP_5}}{\emptyset} = -\varepsilon$ and $\SDD{\emptyset}{\SP_3,\SP_4} = -\varepsilon$, and integrating these partially yields the convergent integrand
	\begin{equation*}
		\frac{(3\varepsilon-2)(3\varepsilon-1)}{2 \varepsilon^3}
		\frac{
			\SP_5^2
			(\SP_3+\SP_4)
			(\SP_1+\SP_2)
		}{
			\phipol^{1+2 \varepsilon}
			\psipol^{4-3 \varepsilon}
		}
		\big[
			2(\SP_1+\SP_2)(\SP_3+\SP_4)
			+(3\varepsilon-1) \SP_5 (\SP_1 + \SP_2 + \SP_3 +\SP_4)
		\big]
	\end{equation*}
	which we used in the computations of \eqref{eq:diagbox-massless-2} and \eqref{eq:diagbox-massless-1}. But note that \eqref{eq:diagbox-massless} has only a pole in $\varepsilon$ of second order and indeed we find that the term $\propto \varepsilon^{-3}$ integrates to zero. In contrast, the much better adapted representation \eqref{eq:double-triangle-2dim-representation} has a manifest second order pole in $\varepsilon$ and no divergences in the parametric integral.
	It is considerably more efficient to evaluate.
\end{example}
	Therefore one might try to find more economic ways of generating analytically regularized, convergent integrands (with only $\phipol$ and $\psipol$ raised to non-integer or negative powers). Note however that an integration-by-parts reduction of the parametric integrands as suggested above could also partially solve this problem.

	Apart from these technicalities, conceptually we face the important open question to combinatorially characterize the linearly reducible graphs in the presence of non-trivial dependence on kinematic invariants. 
	As we recalled in section \ref{sec:single-scale}, only in the massless propagator case such a result is available in form of theorem \ref{theorem:vw-3}. 
	It exploits that $\phipol_G = \psipol_{G_{\bullet}}$ where $G_{\bullet}$ denotes $G$ after identifying the two vertices attached to the external momenta and follows from the plethora of identities and factorization formulas among these $\psi$- and the related \emph{Dodgson}-polynomials \cite{Brown:PeriodsFeynmanIntegrals,BrownYeats:SpanningForestPolynomials}.
	But still this covers only a subset of the linearly reducible topologies and we had to explicitly examine the graph polynomials (using a polynomial reduction algorithm) to arrive at theorem \ref{theorem:4loop-massless-propagators}.

	Hence regarding non-trivial kinematics, it will be inevitable to incorporate the new polynomial $\phipol$ and to find analogous factorization properties in order to arrive at combinatorial criteria sufficient for linear reducibility.
	We hope that the plentiful positive examples in this article motivate progress in this direction.
	
	Even further, the examples of section \ref{sec:extending-reducibility} suggest that in some cases we must abandon the original Schwinger parameters and should look for other representations. A systematic study of suitable changes of variables and in particular criteria exhibiting when these can regain linear reducibility is certainly a demanding but worthwhile project.

{\appendix}
\section{Polynomial reduction and linear reducibility}%
\label{sec:linear-reducibility}
	In a parametric representation, we are naturally working with polylogarithmic functions of several variables: To begin with, from expanding \eqref{eq:feynman-rules-parametric} in say $\varepsilon$ we obtain integrands 
	\begin{equation*}
		F 
		\in
		\Q\big[
			\psipol^{-1}, \phipol^{-1}, 
			\log \psipol, \log \phipol, 
			\setexp{\SP_e, \SP_e^{-1}, \log \SP_e}{e \in E},
			\kinematics
		\big].
	\end{equation*}
	Hence these are iterated integrals in the Schwinger- and kinematic variables and we call
	\begin{equation}\label{eq:restricted-symbol-functions}
		\BarObjects \left( S \right) 
		\defas 
		\setexp{\sum_i \frac{f_i}{g_i} \int \dd\log(h_{i,1}) \cdots \int \dd\log(h_{i,j})}{g_i,h_{i,j} \in S \cup \setexp{\SP_e}{e\in E} \cup \kinematics}
	\end{equation}
	(with rational $f_i$) the functions \emph{with symbol in}\footnote{We always allow for $\dd \log \SP_e$ and $\dd \log s$ of kinematic invariants $s \in \kinematics$ and do not write these in $S$.} $S$, according to the \emph{symbol calculus} of \cite{DuhrGanglRhodes:PolygonsAndSymbols,Duhr:HopfAlgebrasCoproductsSymbols} and also following \cite{BognerBrown:SymbolicIntegration}. 
	So in particular $F \in \BarObjects \left( S_{\emptyset} \right)$ for $S_{\emptyset} \defas \set{\psipol,\phipol}$.
	In this language, the essence of the polynomial reduction algorithm of \cite{Brown:TwoPoint} can be stated as
	\begin{lemma}%
		\label{lemma:polynomial-reduction}
		If $h \in \BarObjects \left( S \right)$ and all $f\in S$ are linear $f = A_f \SP_e + B_f$ in $\SP_e$ and $ H = \int_0^{\infty} h\ \dd \SP_e$ converges, then $H \in \BarObjects \left( S_e \right)$ where $S_e$ is the set of irreducible factors\footnote{Here we drop pure constants $c$ (since $\dd \log c = 0$) and monomials.}
 of
		\begin{equation}%
			\label{eq:polynomial-reduction}
			\setexp{A_f,B_f}{f \in S}
			\ \text{and the resultants}\ 
			\setexp{\resultant[\SP_e]{f}{g} \defas A_f B_g - A_f B_g}{f,g \in S}.
		\end{equation}
	\end{lemma}
	\begin{example}\label{ex:diagbox-massless-polynomial-reduction}
		From the parametric integrand $F$ of \eqref{eq:double-triangle-2dim-representation} we read off the polynomials $S_{\emptyset} = \set{u + x + py + s x y, 1+x, 1+y}$ such that $F \in \BarObjects \left( S_{\emptyset} \right)$.
		Using \eqref{eq:polynomial-reduction} we first deduce $\int_0^{\infty} F \dd y\in \BarObjects \left( S_y \right)$ with $S_y = \set{1+x, u+x, p+sx, u+x - p -sx}$ and then apply \eqref{eq:polynomial-reduction} again to obtain 
		$	\iint_0^{\infty} F\ \dd y \dd x 
			\in
			\BarObjects \left( S_{y,x} \right)
		$.
		In fact we reproduce \eqref{eq:diagbox-massless-symbol} since
		\begin{equation*}
			S_{y,x}
			\defas \left( S_x \right)_y
			= \set{u-p, 1-s, 1-u, p-s, u-1-p+s, p-su}.
		\end{equation*}
	\end{example}
	Under the assumptions of this lemma, \cite{Brown:TwoPoint} describes an entirely combinatorial-algebraic algorithm to effectively compute the integral $\int_0^{\infty} h\ \dd\SP_e$. 
	Let us stress that in particular it does not need any numeric evaluations or separate input of boundary values to fix integration constants, which sometimes is a problem for example within the method of differential equations.
	Details of our implementation will be given in the forthcoming publication of our program.

	We therefore formulate the prerequisite for parametric integration as
	\begin{definition}
		$G$ is called \emph{linearly reducible} if for some ordering $e_1,\ldots,e_N$ of its edges there exist sets $S_n\subset \Q[\SP_{n+1},\ldots,\SP_N]$ of polynomials for all $0\leq n<N$ ($S_0 \defas S_{\emptyset}$) such that for any convergent parametric integrand $F \in \BarObjects \left( S_{\emptyset} \right)$ the partial integrals
		\begin{equation}
			f_n \defas \left[
				\prod_{e=1}^{n} \int_0^{\infty} \dd \SP_e
			\right]
			F
			\quad\text{lie in}\quad
			\BarObjects \left( S_n \right)
			\label{eq:partial-integrals}%
			\quad\text{for any $n + 1 < N = \abs{E}$}
		\end{equation}
		and all $g \in S_n $ are linear in $\SP_{n+1}$.
	\end{definition}
	As in example \ref{ex:diagbox-massless-polynomial-reduction}, repeated application of lemma \ref{lemma:polynomial-reduction} can suffice to prove linear reducibility in simple cases (c.f. the \emph{Fubini} algorithm in \cite{Brown:TwoPoint}), but for this article we employed the way more powerful method of \emph{compatibility graphs} that was developed in \cite{Brown:PeriodsFeynmanIntegrals}.
	This algorithm computes for each set $I \subset E$ of edges a set $S_I \subset \Q[\setexp{\SP_e}{e \notin I}]$ of irreducible polynomials such that the partial integrals $f_I \defas \prod_{e \in I} \int_0^\infty \dd\SP_e\ F $ are analytic outside the \emph{Landau variety} $L_I = \bigcup_{g \in S_I} \set{g=0}$ defined in \cite{Brown:PeriodsFeynmanIntegrals}. These sets $S_I$ are typically much smaller than the upper bounds obtained by lemma \ref{lemma:polynomial-reduction} alone.

	\subsection{Hyperlogarithms}%
	\label{sec:hyperlogs}
	The direct integration of iterated integrals of many variables is possible symbolically as shown in \cite{BognerBrown:SymbolicIntegration}, whereas our approach of \cite{Brown:TwoPoint} is to consider the dependence of the integrand on the next integration variable $z = \SP_n$ only, which reduces the function to the one-dimensional integrals of
	\begin{definition}
		For any word $w \in \Sigma^{\times}$ in letters $\setexp{\omega_{\sigma}}{\sigma\in\Sigma}$ over a finite set $0\in\Sigma\subset\C$, the associated \emph{hyperlogarithm} \cite{LappoDanilevsky} is the iterated integral defined by
		\begin{equation}
			\Hyper_{\omega_0^n} (z)
			\defas \frac{\log^n z}{n!}
			\quad\text{for any $n\in\N_0$ and}\quad
			\Hyper_{\omega_{\sigma}w} (z)
			\defas \int_0^z \frac{\dd z'}{z'-\sigma} L_w(z'). %
			\label{eq:Hyper-def}%
		\end{equation}
	\end{definition}
	\begin{remark}
			These functions are analytic and in general multi-valued on $\C \setminus \Sigma$, but uniquely defined upon restriction to $z \in \C \setminus (-\infty,0]$ and $\abs{z}< \min_{0 \neq \sigma\in \Sigma} \abs{\sigma}$. Also called \emph{Goncharov polylogarithms}, we write
			\begin{equation}\label{eq:hyperlog-def}
				G\left( \sigma_1,\ldots,\sigma_n; z \right)
				\defas
				\Hyper_{\sigma_1 \ldots \sigma_n} (z)
				\defas
				\Hyper_{\omega_{\sigma_1}\ldots\omega_{\sigma_n}} (z)
			\end{equation}%
			and can identify them with a special family of multiple polylogarithms:
			For arbitrary $\sigma_1,\ldots,\sigma_r \in \Sigma\setminus\set{0}$, $n_1,\ldots,n_r\in\N$ and $\abs{z}<\min_{1\leq j\leq r} \abs{\sigma_j}$,
		\begin{equation}
			\Hyper_{\omega_0^{n_r-1}\omega_{\sigma_r} \cdots\, \omega_0^{n_2-1}\omega_{\sigma_2}\omega_0^{n_1-1} \omega_{\sigma_1}} (z)
			= (-1)^r
			\Li_{n_1,\ldots,n_r}\left( \frac{\sigma_2}{\sigma_1}, \ldots, \frac{\sigma_r}{\sigma_{r-1}}, \frac{z}{\sigma_r} \right). %
		\end{equation}
	\end{remark}
	This construction and partial fractioning make it obvious that any function
	\begin{equation}\label{eq:hyperlog-algebra}
		f
		\in
		L(\Sigma) 
		\defas 
		\Q\left[z,
						\setexp{\frac{1}{z-\sigma}}{\sigma\in\Sigma},
						\setexp{L_w}{w\in\Sigma^{\times}}
		\right]
	\end{equation}
	in the algebra $L(\Sigma)$ spanned by the hyperlogarithms has a primitive $\partial_z F(z) = f (z)$,
	\begin{equation*}
		F
		\in
		\Q\left[\Sigma \cup \setexp{\frac{1}{\sigma_i-\sigma_j}}{\sigma_i\neq \sigma_j \ \text{from}\ \Sigma}\right]
		\tp 
		L(\Sigma)
	\end{equation*}
	possibly involving the additional denominators $\sigma_i-\sigma_j$. This mirrors lemma \ref{lemma:polynomial-reduction} since
	\begin{equation}
		\label{eq:hyperlog-alphabet-from-symbol}%
			f_I \in L\left( \Sigma_{I,e} \right) \left( \SP_e \right)
			\quad\text{with}\quad
			\Sigma_{I,e}
			\defas
			\set{0}
			\cup
			\bigcup_{f \in S_I} \set{\text{zeros of $f$ with respect to $\SP_e$}}
	\end{equation}
	whenever $f_I \in \BarObjects \left( S_I \right)$. The final answer $f=f_{\abs{E}-1}$ after integrating out all Schwinger variables (but $\SP_{\abs{E}}=1$) has a symbol in $f \in \BarObjects \left( S_{\abs{E}} \right)$.

	To represent a function $ f \in \BarObjects \left( S \right)$ in terms of one-dimensional iterated integrals we choose an order $s_1,\ldots,s_n$ of the remaining variables $\kinematics$ and express $f$ in the form\footnote{This corresponds to fixing the path of integration in the iterated integral to the piecewise linear $(0,\ldots,0) \rightarrow (0,\ldots,0,s_n) \rightarrow (0,\ldots,0,s_{n-1},s_n) \rightarrow \ldots \rightarrow (0,s_2,\ldots,s_n) \rightarrow (s_1,\ldots,s_n)$. The discussion in section 2.7 of \cite{HennSmirnov:Bhabha} might further clarify this process.}
	 \begin{equation}%
		 \label{eq:fibration-basis}
		 f
		 \in
		 L\left( \Sigma_{s_1} \right) ( s_1 )
		 \tp
		 \ldots
		 \tp
		 L\left( \Sigma_{s_n} \right) ( s_n )
	 \end{equation}
	 where $\Sigma_{s_i} = \set{0} \cup \set{\text{zeros of $S^{(i)}$ w.r.t. $s_i$}}$ and we set $S^{(1)} \defas S$ and recursively define $S^{(i+1)}$ as the irreducible factors\footnote{Again we omit constants and monomials as these are explicitly taken care of by $0\in \Sigma_{s_i}$.} of $\lim_{s_i \rightarrow 0} S^{(i)}$.
	\begin{example}\label{ex:diagbox-massless-fibration-basis}
		If we chose the order $0 < s \ll u \ll p$ of variables, we can write the integrals $f$ of the expansion coefficients $F$ of the integrand \eqref{eq:double-triangle-2dim-representation} in the form
		\begin{equation}%
			\label{eq:diagbox-massless-fibration-basis}
			f
			\in
			L\left( \set{0, 1, p, 1+p-u,\frac{p}{u}} \right) (s)
			\tp
			L\left( \set{0, 1, p, 1+p} \right) (u)
			\tp
			L\left( \set{0, -1} \right) (p)
		\end{equation}
		by taking $S^{(1)}$ from \eqref{eq:diagbox-massless-symbol} and deducing $S^{(2)} = \lim_{s\rightarrow 0} S^{(1)} = \set{u-p,1-u,u-1-p}$ and $S^{(3)} = \lim_{u\rightarrow 0} S^{(2)} = \set{1+p}$. 
		Indeed we find precisely the letters given in \eqref{eq:diagbox-massless-fibration-basis} in our results like \eqref{eq:diagbox-massless-2}, \eqref{eq:diagbox-massless-1}.
	\end{example}
	\begin{example}
		The final set $S_{\abs{E}} = \Sigma_{\Delta} \cup \set{z\bar{z}-1}$ in the polynomial reduction of the linearly reducible massless off-shell three-point graph $\Delta_{3,5}$ of figure \ref{fig:vertices-3loops} shows that
		\begin{equation}%
			\label{eq:v3_5-fibration-basis}
			\Phi\left( \Delta_{3,5} \right)
			=
			\frac{\Gamma(2+3\varepsilon)}{p_1^{4+6\varepsilon}} \sum_{n=0}^{\infty} f_n \varepsilon^n
			\quad\text{with}\quad
			f_n
			\in
			L\left( \set{ 0, 1, \bar{z}, \frac{1}{\bar{z}}} \right) (z)
			\tp
			L\left( \set{ 0, 1} \right) (\bar{z}).
		\end{equation}
		Results for $f_0$ and $f_1$ are supplied (in this form) in the attached file.
	\end{example}

\phantomsection
\pdfbookmark[1]{References}{final-bibliography}
\bibliographystyle{JHEP}
\bibliography{../qft}

\end{document}